\documentclass{CSML}

\def\dOi{12(4:8)2016}
\lmcsheading%
{\dOi}
{1--34}
{}
{}
{Apr.~\phantom01, 2016}
{Dec.~28, 2016}
{}

\ACMCCS{[{\bf Theory of computation}]: Computational complexity and
  cryptography---Circuit complexity\,/\,Complexity calsses; Logic---Finite Model Theory;
  Theory and algorithms for application domains---Database
  theory---Database query languages (principles); }

\usepackage{picins}
\usepackage{microtype}
\usepackage{amssymb,amsmath,amsthm,graphicx,xspace,calc}
\usepackage{ifthen}
\usepackage{latexsym}
\usepackage{hyperref}
\usepackage[utf8]{inputenc}
\usepackage{url}
\usepackage{stmaryrd,mathrsfs}

\renewcommand{\phi}{\varphi}
\renewcommand{\geq}{\ensuremath{\geqslant}}
\renewcommand{\leq}{\ensuremath{\leqslant}}

\newcommand{\sphere}[3][k]{\ensuremath{(\NS^{#2}_{#1}(#3)},#3)}

\newcommand{\id}[1]{\operatorname{id}_{#1}}

\newenvironment{proofc}[1]{\begin{proof}[Proof of #1]}{\end{proof}}

\theoremstyle{plain}
\newtheorem{theorem}{Theorem}[section]
\newtheorem{lemma}[theorem]{Lemma}
\newtheorem{corollary}[theorem]{Corollary}
\newtheorem{proposition}[theorem]{Proposition}
\newtheorem*{claim}{Claim}

\theoremstyle{definition}
\newtheorem{definition}[theorem]{Definition}
\newtheorem{example}[theorem]{Example}
\newtheorem{remark}[theorem]{Remark}

\newcommand{\myparagraph}[1]{\noindent\textbf{\large #1. }}

\newcommand{\abs}[1]{\vert{}#1\vert{}}

\newcommand{\set}[1]{\ensuremath{\{ #1 \}}}
\newcommand{\setc}[2]{\set{#1 \,:\, #2}}

\newcommand{\isom}{\ensuremath{\cong}}
\newcommand{\Hanfequiv}{\ensuremath{\leftrightarrows}}

\newcommand{\dist}{\ensuremath{\textit{dist}}}

\newcommand{\ov}[1]{\ensuremath{\overline{#1}}}
\newcommand{\mymod}{\ensuremath{\textup{ mod }}}

\newcommand{\NN}{\ensuremath{\mathbb{N}}}
 
\newcommand{\NNpos}{\ensuremath{\NN_{\mbox{\tiny $\scriptscriptstyle \geq 1$}}}}

\newcommand{\union}{\ensuremath{\cup}}

\newcommand{\und}{\ensuremath{\wedge}}
\newcommand{\Und}{\ensuremath{\bigwedge}}
\newcommand{\oder}{\ensuremath{\vee}}
\newcommand{\Oder}{\ensuremath{\bigvee}}
\newcommand{\nicht}{\ensuremath{\neg}}

\newcommand{\existsc}[1]{\ensuremath{\exists^{0 \textup{ mod } #1}}}
\newcommand{\existscc}[2]{\ensuremath{\exists^{#1 \textup{ mod } #2}}}

\newcommand{\free}{\ensuremath{\textit{free}}}
\newcommand{\ar}{\ensuremath{\textit{ar}}}

\newcommand{\ARB}{\ensuremath{\textrm{arb}}}

\newcommand{\Rep}{\mathrm{Rep}}
\newcommand{\pat}{\mathit{pattern}}

\newcommand{\struct}[1]{\ensuremath{\mathcal{#1}}}
\newcommand{\A}{\struct{A}}
\newcommand{\B}{\struct{B}}
\newcommand{\C}{\struct{C}}
\newcommand{\G}{\struct{G}}
\renewcommand{\S}{\struct{S}}
\newcommand{\T}{\struct{T}}
\newcommand{\HS}{\struct{T}}
\newcommand{\NS}{\struct{N}}

\newcommand{\rela}[2]{\ensuremath{#1\!\upharpoonright\!#2}}

\newcommand{\Class}{\ensuremath{\mathfrak{C}}}
\newcommand{\ClassStrings}[1][\Sigma]{\ensuremath{#1\textit{-strings}}}

\newcommand{\stringfont}[1]{\texttt{#1}}

\newcommand{\logic}[1]{\ensuremath{\textrm{\upshape{#1}}}}
\newcommand{\FO}{\logic{FO}}
\newcommand{\FOMOD}[1]{\ensuremath{\textrm{\upshape FO+MOD}_{#1}}}
\newcommand{\FOMODp}{\ensuremath{\FOMOD{p}}}

\newcommand{\arbinv}{\ensuremath{\textup{arb-inv-}}}
\newcommand{\arbinvFO}{\ensuremath{\textup{arb-inv-}\FO}}
\newcommand{\ordinvFO}{\ensuremath{{<}\textup{-inv-}\FO}}

\newcommand{\arbinvFOMODp}{\ensuremath{\textup{arb-inv-$\FOMODp$}}}
\newcommand{\ordinvFOMODp}{\ensuremath{{<}\textup{-inv-$\FOMODp$}}}

\newcommand{\arbinvFOMODpC}{\ensuremath{\textup{arb-inv-$\textrm{\upshape FO+MOD}_{p}^{\Class}$}}}

\newcommand{\ordinvFOMODnum}[1]{\ensuremath{{<}\textup{-inv-$\FOMOD{#1}$}}}

\newcommand{\compclass}[1]{\ensuremath{\textrm{\upshape{#1}}}}

\newcommand{\ACO}{\ensuremath{\compclass{AC}^0}}
\newcommand{\MOD}[1]{\ensuremath{\textup{MOD}_{#1}}}
\newcommand{\MODp}{\MOD{p}}

\newcommand{\ACCO}{\ensuremath{\compclass{ACC}^0}}
\newcommand{\NEXP}{\compclass{NEXP}}
\newcommand{\NLOGSPACE}{\compclass{NLOGSPACE}}
\newcommand{\PTIME}{\compclass{PTIME}}

\newcommand{\twist}{\mathop{twist}}
\newcommand{\tturn}{\mathop{turn}}

\newcommand{\occ}{\operatorname{occ}}

\begin{document}

\title[Arb-invariant first-order formulas with
modulo counting quantifiers]{On the locality of arb-invariant
  first-order formulas with modulo counting quantifiers\rsuper*}
\titlecomment{{\lsuper*}This is the full version of the conference
  contribution \cite{HarwathSchweikardt2013}}

      \author[F.~Harwath]{Frederik Harwath\rsuper a} \address{{\lsuper
          a}Institut f\"ur Informatik,
        Goethe-Universit\"at Frankfurt, Germany}
      \email{harwath@cs.uni-frankfurt.de}
  
      \author[N.~Schweikardt]{Nicole Schweikardt\rsuper b}
      \address{{\lsuper b}Institut f\"ur Informatik,
        Humboldt-Universit\"at zu Berlin, Germany}
      \email{schweikn@informatik.hu-berlin.de}


\begin{abstract}
  We study Gaifman locality and Hanf locality of an extension of
  first-order logic with modulo $p$ counting quantifiers ($\FOMODp$,
  for short) with arbitrary numerical predicates. We require that the
  validity of formulas is independent of the particular interpretation
  of the numerical predicates and refer to such formulas as
  arb-invariant formulas.  This paper gives a detailed picture of
  locality and non-locality properties of arb-invariant $\FOMODp$.
  For example, on the class of all finite structures, for any
  $p\geq 2$, arb-invariant $\FOMODp$ is neither Hanf nor Gaifman local
  with respect to a sublinear locality radius. However, in case that
  $p$ is an odd prime power, it is \emph{weakly} Gaifman local with a
  polylogarithmic locality radius. And when restricting attention to
  the class of string structures, for odd prime powers $p$,
  arb-invariant $\FOMODp$ is both Hanf and Gaifman local with a
  polylogarithmic locality radius. Our negative results build on
  examples of order-invariant $\FOMODp$ formulas presented in
  Niemist\"o's PhD thesis. Our positive results make use of the close
  connection between $\FOMODp$ and Boolean circuits built from
  NOT-gates and AND-, OR-, and MOD$_p$- gates of arbitrary fan-in.
\end{abstract}

\maketitle

\section{Introduction}\label{section:introduction}
Expressivity of logics over finite structures plays an important
role in various areas of computer science. In descriptive complexity,
logics are used to characterise complexity classes, and concerning
databases, common query languages have well-known logical equivalents.
These applications have motivated a systematic study of the expressive
powers of logics on finite structures. The classical inexpressibility
arguments for logics over finite structures (i.e., back-and-forth
systems or Ehrenfeucht-Fra\"\i{}ss\'{e} games; cf.\ 
\cite{Libkin2004}) often involve nontrivial combinatorics. Notions of
\emph{locality} have been proposed as an alternative that allows to
contain much of the hard combinatorial work in generic results.

The two best known
notions of locality are \emph{Gaifman locality} and \emph{Hanf locality},
introduced in \cite{HellaLibkinNurmonen1999,FSV95}.
A $k$-ary query is called \emph{Gaifman local with locality radius
$f(n)$} if in a structure of cardinality $n$, the question whether a
given tuple satisfies the query only depends on the isomorphism type
of the tuple's neighbourhood of radius $f(n)$.
A Boolean query is \emph{Hanf local with locality radius
$f(n)$} if the question whether a structure of size $n$ satisfies the
query only depends on the number of occurrences of isomorphism types
of neighbourhoods of radius $f(n)$.
If a given logic is capable
of defining only Gaifman or Hanf local queries with a sublinear
locality radius, then this logic cannot express ``non-local'' queries
such as, e.g., the query asking whether two nodes of a graph are
connected by a path, or the query asking whether a graph is acyclic
(cf., e.g., the textbook \cite{Libkin2004}).
It is well-known that first-order logic $\FO$, as well as extensions
of $\FO$ by various kinds of counting quantifiers, are Gaifman local and
Hanf local with a constant locality radius \cite{HellaLibkinNurmonen1999,FSV95}.
Also, locality properties of extensions of $\FO$ by invariant uses of order and arithmetic
have been considered \cite{Grohe1998,AMSS-SICOMP}.

Order-invariant and arb-invariant formulas
were introduced to capture
the data independence principle in databases: An implementation of a
database query may exploit the order in which the database elements
are stored in memory, and thus identify the elements with natural
numbers on which arithmetic can be performed. But the use of order and
arithmetic should be restricted in such a way that the result of the
query does not depend on the particular order in which the data is
stored.
In addition to the relations of a given structure, \emph{arb-invariant} formulas are allowed to use a linear order $<$ and arithmetic predicates such as $+$ or $\times$ induced by the order.
But the truth of such a formula in a structure must be independent of the particular linear order which is chosen for $<$. Arb-invariant formulas that only use the linear order, but no further arithmetic predicates, are called \emph{order-invariant}.
In \cite{Grohe1998} it was shown that order-invariant $\FO$ can express
only queries that are Gaifman local with a constant locality
radius, and from \cite{AMSS-SICOMP} we know that arb-invariant $\FO$ can
express only queries that are Gaifman local with a polylogarithmic
locality radius.
The proof of \cite{AMSS-SICOMP} relies on a reduction using strong lower
bound results from circuit complexity, concerning
$\ACO$-circuits. 
Similar lower bounds are known also for the
extension of $\ACO$-circuits by modulo $p$ counting gates, for a prime
power $p$ \cite{Smolensky}. This naturally raises the question
whether the locality results from \cite{AMSS-SICOMP} can be generalised
to the extension of $\FO$ by modulo $p$ counting quantifiers
($\FOMODp$, for short), which precisely corresponds to $\ACO$-circuits
with modulo $p$ counting gates \cite{BIS90}. 
This question was the starting point
for the investigations carried out in the present paper. \par Our results
give a detailed picture of the locality and non-locality properties of
order-invariant and arb-invariant $\FOMODp$:
For every natural number $p\geq 2$, order-invariant $\FOMODp$ is
neither Hanf nor Gaifman local with a sublinear locality radius
(see Section~\ref{section:strings} and Proposition~\ref{prop:notGaifmanLocal}).
For \emph{even} numbers $p\geq 2$, order-invariant $\FOMODp$ is
not even \emph{weakly} Gaifman local with a sublinear locality radius
(Proposition~\ref{prop:notWeaklyGaifmanLocal}). 
Here, \emph{weak} Gaifman locality is a relaxed notion of Gaifman
locality referring only to tuples with disjoint neighbourhoods (cf.,
\cite{Libkin2004}).
However, for \emph{odd prime powers} $p$ we can show that
arb-invariant $\FOMODp$ is weakly Gaifman local with a
polylogarithmic locality radius (Theorem~\ref{thm:weakGaifmanLocality}). 
For showing the latter result, we introduce a new locality notion
called \emph{shift locality}, for which we can prove for all prime
powers $p$ that arb-invariant $\FOMODp$ is shift local with a
polylogarithmic locality radius
(Theorem~\ref{thm:CyclicShiftLocality}). 
Our proof relies on Smolensky's
circuit lower bound \cite{Smolensky}. Generalising our result from
prime powers $p$ to arbitrary numbers $p$ can be expected to be
difficult, since it would solve long-standing open questions in
circuit complexity (see Remark~\ref{remark:circuitcomplexity}).
\par
When restricting attention to the class of \emph{string structures}, we
obtain for odd prime powers $p$, that arb-invariant $\FOMODp$ is both
Hanf and Gaifman local with a polylogarithmic locality radius
(Theorem~\ref{thm:HanfLocalityOnStrings} and Corollary~\ref{cor:GaifmanLocalityOnStrings}).
On the other hand, for even numbers $p\geq 2$, order-invariant
$\FOMODp$ on string structures is neither Gaifman nor Hanf local
with a sublinear locality radius
(Proposition~\ref{prop:notWeaklyGaifmanLocal} and Section~\ref{section:strings}). 
This, in particular, implies
that order-invariant $\FOMODp$ is strictly more expressive on strings
than $\FOMODp$, refuting a conjecture of Benedikt and Segoufin \cite{BenediktSegoufin2009a}.

\medskip

\myparagraph{Outline}
Section~\ref{section:preliminaries} fixes the basic notation,
introduces order-invariance and arb-invariance,
and recalls two examples of order-invariant $\FOMODp$-formulas from
Niemist\"o's PhD-thesis
\cite{Niemistoe-PhD}. Section~\ref{section:gaifman} presents our
results concerning Gaifman locality, weak Gaifman locality, and shift
locality. Section~\ref{section:strings} deals with Gaifman and Hanf locality on string structures. Section~\ref{section:conclusion} concludes the paper.

\smallskip

\myparagraph{Acknowledgements} We want to thank the anonymous reviewers for their valuable suggestions.

\section{Preliminaries}\label{section:preliminaries}
We assume that the reader is familiar with the basic concepts and
notations concerning first-order logic and extensions thereof
which can be found in most modern introductory books on mathematical logic.
In particular, we use the standard semantics of first-order logic.
The exposition given in the textbooks \cite{Libkin2004,Ebbinghaus1999} is
close to our concerns. These textbooks also provide the reader with the necessary background
on the relation of first-order logic to circuit complexity, on notions of locality, and on order-invariant first-order formulas.\footnote{These topics are more emphasised in \cite{Libkin2004} which, in particular, contains several examples for the different notions of locality.}

\myparagraph{Basic notation}
We write $\NN$ for the set of non-negative integers and
let $\NNpos:=\NN\setminus\set{0}$. 
For $n\in \NNpos$ we write $[n]$ for the set
$\setc{i\in\NN}{0\leq i<n}$, i.e., $[n]=\set{0,\ldots,n{-}1}$.
For integers $i,j,p$ with $p\geq 2$, we write $i\equiv
j \mymod p$ (and say that $i$ is congruent $j$ modulo $p$) iff
there exists an integer $k$ such that $i=j+kp$. For integers $i,i'$, the
term $(i{+}i'\mymod p)$ denotes the number $j\in [p]$ such that
$i{+}i'\equiv j \mymod p$.
A number $p$ is called a \emph{prime power} if
$p=\hat{p}^i$ for a prime $\hat{p}$ and an integer $i\geq 1$,
and $p$ is called an
\emph{odd} prime power if $p$'s prime factor is different from $2$
(i.e., $p$ is odd).
We write $\log n$ to denote
the logarithm of a number $n$ with respect to base 2, and we often
simply write $\log n$ instead of $\lfloor \log n \rfloor$.

By $2^A$ we denote the power set of $A$, i.e., the set
$\setc{Y}{Y\subseteq A}$.
The set of all non-empty finite strings built from symbols in $A$ is
denoted $A^+$. We write $|w|$ for the length of a string $w\in A^+$.
For an $a\in A$ we write $|w|_a$ for the number of occurrences
of the letter $a$ in the string $w$.

\bigskip

\myparagraph{Structures}
A \emph{signature} $\sigma$ is a set of relation symbols $R$, each of
them associated with a fixed arity $\ar(R)\in \NNpos$.
Throughout this paper, $\sigma$ will usually denote a fixed finite signature.

A \emph{$\sigma$-structure} $\A$ consists of a non-empty set $A$
called the \emph{universe} of $\A$, and a relation $R^\A\subseteq
A^{\ar(R)}$ for each relation symbol $R\in\sigma$.
The \emph{cardinality} of a $\sigma$-structure $\A$ is the cardinality
of its universe. 
\emph{Finite} $\sigma$-structures are $\sigma$-structures of finite cardinality.
For $\sigma$-structures $\A$ and $\B$ and tuples $\ov{a}=(a_1,\ldots,a_k)\in A^k$ and
$\ov{b}=(b_1,\ldots,b_k)\in B^k$ we write $(\A,\ov{a})\isom (\B,\ov{b})$ to indicate
that there is an isomorphism $\pi$ from $\A$ to $\B$ that maps
$\ov{a}$ to $\ov{b}$ (i.e., $\pi(a_i)=b_i$ for each $i\leq k$).

We represent strings over a finite alphabet $\Sigma$ by successor-based
structures as follows:
We choose $\sigma_\Sigma:=\set{E}\cup\setc{P_a}{a\in\Sigma}$, where $E$ is a
binary relation symbol and $P_a$ is a unary relation symbol, for each
$a\in\Sigma$.
We represent a non-empty string $w\in\Sigma^+$ by the
$\sigma_\Sigma$-structure $\S_w$, where the universe of
$\S_w$ is the set $\set{1,\ldots,|w|}$ of positions of $w$, the edge
relation $E^{\S_w}$ is the successor relation, i.e.,
$E^{\S_w}=\setc{(i,i+1)}{1\leq i<|w|}$, and for each $a\in\Sigma$,
the set $P_a^{\S_w}$ consists of all positions of $w$ that carry the
letter $a$.
Structures of the form $\S_w$ (for a string $w$) are called \emph{string
  structures}. Note that in the literature, strings are also often represented by linearly ordered structures,
whereas our representation contains only the successor relation, but not the linear order.

In this paper, we restrict attention to \emph{finite} structures. All
\emph{classes} $\Class$ of $\sigma$-structures will be closed
under isomorphism, i.e., if $\A$ and $\B$ are isomorphic
$\sigma$-structures, then $\A\in\Class$ iff $\B\in\Class$.
We will write $\ClassStrings$ to denote the class of all
$\sigma_\Sigma$-structures that represent strings in $\Sigma^+$ (i.e.,
$\ClassStrings$ is the closure under isomorphisms
of the set $\setc{\S_w}{w\in\Sigma^+}$).

\bigskip

\myparagraph{First-order logic with modulo counting quantifiers}
By $\free(\varphi)$ we denote the set of all free variables of a formula
$\varphi$. A \emph{sentence} is a formula $\varphi$ with
$\free(\varphi)=\emptyset$.
We often write $\varphi(\ov{x})$, for
$\ov{x}=(x_1,\ldots,x_k)$, to indicate that
$\free(\varphi)=\set{x_1,\ldots,x_k}$.  
If $\A$ is a $\sigma$-structure and $\ov{a}=(a_1,\ldots,a_k)\in A^k$, we write
$\A\models\varphi[\ov{a}]$ to indicate that the formula
$\varphi(\ov{x})$ is satisfied in $\A$ when interpreting the free
occurrences of the variables $x_1,\ldots,x_k$ with the elements
$a_1,\ldots,a_k$. 

We write $\FO(\sigma)$ to denote the class of all
first-order formulas of signature $\sigma$.
In this paper, we consider the extension of $\FO(\sigma)$ by \emph{modulo counting quantifiers},
defined as follows:
Let $p$ be a natural number with $p\geq 2$. 
A \emph{modulo $p$ counting quantifier} is of the form
$\existscc{i}{p}$ for some $i\in [p]$. A formula of the form $\existscc{i}{p}x\,\varphi(x,\ov{y})$ is satisfied by a finite $\sigma$-structure $\A$ and an interpretation $\ov{b}\in A^k$ of the variables $\ov{y}$ iff the number of elements $a\in A$ such that $\A\models\varphi[a,\ov{b}]$ is congruent $i$ modulo $p$. 

For a fixed natural number $p\geq 2$ 
we write $\FOMODp(\sigma)$ to denote the extension of $\FO(\sigma)$ by modulo $p$ counting quantifiers. That is, $\FOMODp(\sigma)$ is built from atomic formulas of the form $x_1{=}x_2$ and
$R(x_1,\ldots,x_{\ar(R)})$, for $R\in\sigma$ and variables
$x_1,x_2,\ldots,x_{\ar(R)}$, and closed under Boolean connectives
$\und$, $\oder$, $\nicht$, existential and universal first-order
quantifiers $\exists$, $\forall$, and modulo $p$ counting quantifiers
$\existscc{i}{p}$, for $i\in [p]$.
This logic has been studied in depth, see e.g.,
\cite{Straubing1994,HellaLibkinNurmonen1999,BIS90}.
Note that if $m$ is a multiple of $p$, then $\FOMOD{m}$ can express
modulo $p$ counting quantifiers, since \ $\existscc{i}{p}x\
\varphi(x,\ov{y})$ \ is
equivalent to \ $\Oder_{0\leq j < m/p}\existscc{jp+i}{m}x\ \varphi(x,\ov{y})$.

\bigskip

\myparagraph{Arb-invariant formulas}
We can extend the expressive power of a logic by allowing formulas to use,
apart from the relation symbols present in the signature $\sigma$, also 
a linear order $<$, arithmetic predicates such as $+$ or $\times$, or
arbitrary numerical predicates. We define an $r$-ary \emph{numerical predicate} $P^\NN$ as an $r$-ary relation on $\NN$
(i.e., $P^\NN\subseteq \NN^r$). 
Note that, compared to the definition of numerical predicates used, e.g., in \cite{Straubing1994}, this definition simplifies some details, while
it is essentially equivalent. That is, each $r$-ary numerical predicate according to the definition of \cite{Straubing1994}
can be represented by an $(r{+}1)$-ary numerical predicate according to our definition.
Two examples of numerical predicates are
the linear order 
$<^\NN$ consisting of all tuples $(a,b)\in\NN^2$ with $a<b$, and
the addition predicate $+^\NN$ consisting of all triples
$(a,b,c)\in\NN^3$ with $a+b=c$.

To allow formulas to use numerical predicates, we fix
the following notation: For every $r\in\NNpos$ and every $r$-ary
numerical predicate $P^\NN$, let $P$
be a new relation symbol of arity $r$ (``new'' meaning that $P$
does not belong to $\sigma$). 
We write $\eta_\ARB$ to denote the set of all the relation symbols $P$
obtained this way, and let $\sigma_\ARB:=\sigma\cup\eta_\ARB$
(the subscript ``{\ARB}'' stands for ``arbitrary numerical predicates'').

Next, we would like to allow $\FOMODp(\sigma_\ARB)$-formulas to make
meaningful statements about finite $\sigma$-structures. To this end,
for a finite $\sigma$-structure $\A$, 
we consider embeddings $\iota$ of the universe of $\A$
into the initial segment of $\NN$ of size $n=|A|$, i.e., the set 
$[n]=\set{0,\ldots,n{-}1}$.

\begin{definition}[Embedding]
 Let $\A$ be a finite $\sigma$-structure, and let $n:=|A|$.\\
 An \emph{embedding} $\iota$ of $\A$ is a bijection $\iota:A\to [n]$.
\end{definition}

Given a finite $\sigma$-structure $\A$ and an embedding $\iota$ of
$\A$,  we can translate $r$-ary numerical predicates
$P^\NN$ into $r$-ary predicates on $A$ as follows:
$P^\NN$ induces
the $r$-ary predicate $P^\iota$ on $A$, consisting of all $r$-tuples 
$\ov{a}=(a_1,\ldots,a_r)\in A^r$ where
$\iota(\ov{a})=(\iota(a_1),\ldots,\iota(a_r))\in P^\NN$. 
In particular, the linear order $<^\NN$ induces the linear order $<^\iota$ on $A$ where
for all $a,b\in A$ we have $a<^\iota b$ iff $\iota(a)<^{\NN} \iota(b)$.

The $\sigma_\ARB$-structure $\A^\iota$ associated with $\A$ and
$\iota$ is the expansion of $\A$ by the predicates $P^\iota$ for all
$P\in\eta_\ARB$. That is, $\A^\iota$ has the same universe as $\A$, all
relation symbols $R\in\sigma$ are interpreted in $\A^\iota$ in the
same way as in $\A$, and every numerical symbol $P\in\eta_\ARB$ is interpreted
by the relation $P^\iota$.

To ensure that an $\FOMODp(\sigma_\ARB)$-formula $\varphi$ makes a meaningful
statement about a $\sigma$-structure $\A$, we evaluate $\varphi$ in
$\A^\iota$, and we restrict attention to those
formulas whose truth value is independent of the particular choice of
the embedding $\iota$. This is formalised by the following notion.

\begin{definition}[Arb-invariance]
Let $\varphi(\ov{x})$ be
an $\FOMODp(\sigma_\ARB)$-formula with $k$ free variables, and
let $\A$ be a finite $\sigma$-structure.
The formula $\varphi(\ov{x})$ is \emph{arb-invariant on $\A$} if for
\emph{all} embeddings $\iota_1$ and $\iota_2$ of $\A$ and for all
tuples $\ov{a}\in A^k$ we have: \
$
   \A^{\iota_1} \models \varphi[\ov{a}] 
   \iff 
   \A^{\iota_2} \models \varphi[\ov{a}].
$

Let $\varphi(\ov{x})$ be arb-invariant on $\A$. 
We write
$\A\models\varphi[\ov{a}]$, if $\A^\iota\models\varphi[\ov{a}]$ for
some (and hence every) embedding $\iota$ of $\A$.
\end{definition}

\begin{definition}[$\arbinvFOMODp$]
An $\FOMODp(\sigma_\ARB)$-formula
$\varphi(\ov{x})$ is \emph{arb-invariant on a class $\Class$} of
finite $\sigma$-structures, if $\varphi(\ov{x})$ is
arb-invariant on every $\A\in\Class$.
We denote the set of all
$\FOMODp(\sigma_\ARB)$-formulas that are arb-invariant on $\Class$ by $\arbinvFOMODpC(\sigma)$.

A formula $\varphi(\ov{x})$ is called \emph{arb-invariant} if it is
arb-invariant on the class of all finite $\sigma$-structures.
We write $\arbinvFOMODp(\sigma)$ to denote the set of all
arb-invariant $\FOMODp(\sigma_\ARB)$-formulas.
\end{definition}

\begin{definition}[Order-invariance and $\ordinvFOMODp$]\ \\
An arb-invariant formula that only uses the numerical predicate $<^\NN$ is called \emph{order-invariant}.
\\
By $\ordinvFOMODp(\sigma)$ we denote the set of all arb-invariant $\FOMODp(\sigma\cup\set{<})$-formulas.
\end{definition}

Next, we present two examples of $\ordinvFOMODp(\sigma)$-sentences that were developed by Niemist\"o in \cite{Niemistoe-PhD} and that will be used later on in this paper
as examples for the locality and non-locality properties of $\arbinvFOMODp(\sigma)$-sentences.
\begin{example}[Niemist\"o (Proposition 6.22 in
  \cite{Niemistoe-PhD})]\label{example:niemistoe-evencycles}

Let $\sigma=\set{E}$ be the signature consisting of a binary relation symbol $E$. This example presents 
an $\ordinvFOMODnum{2}(\sigma)$-sentence $\varphi_{\textit{even
    cycles}}$ that is satisfied by exactly those finite
$\sigma$-structures $\A$ that are disjoint unions of directed cycles
where the number of cycles of even length is even. 

Clearly, a finite $\sigma$-structure $\A$ is a disjoint union of
directed cycles iff every element $a\in A$ has in-degree 1 and
out-degree 1 --- and this can easily be expressed by an
$\FO(\sigma)$-sentence $\varphi_{\textit{cycles}}$.

Each $\A$ with $\A\models\varphi_{\textit{cycles}}$ can be identified
with the permutation $\pi_\A$ of $A$ where, for every $a\in A$,
$\pi_\A(a)=b$ for the unique element $b\in A$ with $(a,b)\in E^\A$.  
Note that the cycles of $\A$ precisely correspond to the cycle
decomposition of the permutation $\pi_\A$.  

Now let $n:=|A|$, let $\iota$ be an arbitrary embedding of $\A$ into
$[n]$, and let $\pi_\A^\iota$ be defined via $\pi_\A^\iota(j):=
\iota(\pi_\A(\iota^{-1}(j)))$, for every $j\in [n]$.  
Clearly, $\pi_\A^\iota$ is a permutation of $[n]$, and the cycle
decomposition of $\pi_\A^\iota$ is obtained from the cycle
decomposition of $\pi_\A$ by replacing every $a\in A$ with the number
$\iota(a)$. 

It is a well-known fact concerning permutation groups
(see \cite{Rotman-GroupTheory})
that the cycle decomposition of $\pi_\A^\iota$ has an even number of
cycles of even length if, and only if, the number of
\emph{inversions}, i.e., pairs $(x,y)\in [n]^2$ with $x<y$ and
$\pi_\A^\iota(y)<\pi_\A^\iota(x)$, is even. 
The latter can easily be expressed in $\FOMOD{2}(\sigma\cup\set{<})$
by the formula 
\[
  \varphi_{\textit{inversions}} \ := \ \ 
  \existscc{0}{2}x\ \ \existscc{1}{2}y \ \ \big(\
    x\,{<}\,y \ \und \ \ \pi(y)\,{<}\,\pi(x)
  \ \big),
\]
where $\pi(y)\,{<}\,\pi(x)$ is an abbreviation for the formula
$\exists x'\exists y' \big( E(x,x')\und E(y,y') \und \
y'\,{<}\,x'\big)$. 

In summary, the formula \ $\varphi_{\textit{even
    cycles}}:=(\varphi_{\textit{cycles}}\und\varphi_{\textit{inversions}})$
\ is an $\ordinvFOMODnum{2}(\sigma)$-sentence that is satisfied by
exactly those finite $\sigma$-structures that are disjoint unions of
cycles where the number of cycles of even length is even. 
\qed
\end{example}

\begin{example}[Niemist\"o (Proposition 6.20 in
  \cite{Niemistoe-PhD})]\label{example:niemistoe-torus}
  Let $\sigma=\set{E_1,E_2}$ be the signature consisting of two binary
  relation symbols $E_1$ and $E_2$. Let $h \in \NN$. A \emph{torus} is
  a $\sigma$-structure $\T$ with universe $[h]\times [w]$, for some $w
  \geq 2$, and relations
  \begin{eqnarray*}
    E_1^{\T} & := & 
    \setc{\ \big((i,j),\,(i{+}1\mymod h,\,j)\big)\ }{\ i \in [h],\ \
      j\in [w]\ }
    \\[0.5ex]
    E_2^{\T} & := & \setc{\;\big((i,j),\,(i,j{+}1)\big)\;}{\;i\in
      [h],\ j\in [w{-}1]\;} \ \cup \\
    & &\setc{\ \big((i,w{-}1)\,,\,(i{+}k \mymod h,\, 0)\big)\ }{\;i\in [h]\;} \ , \  \text{ for some $k\in [h]$. }
  \end{eqnarray*}
  The number $w$ is the \emph{width} of $\T$. The set $[h] \times
  \set{j}$ is the \emph{$j$-th column} of $\T$, for each $j\in [w]$.
  The number $h$ is the \emph{height} of $\T$ which will be fixed
  throughout this example. The number $k$ is the \emph{twist} of $\T$
  which we also denote by $\twist(\T)$. See Figure~\ref{figure:torus-and-twistedtorus}
  for an illustration of two tori with different twist.
  \begin{figure}[ht]
    \centering\vspace{-6ex}
    \includegraphics{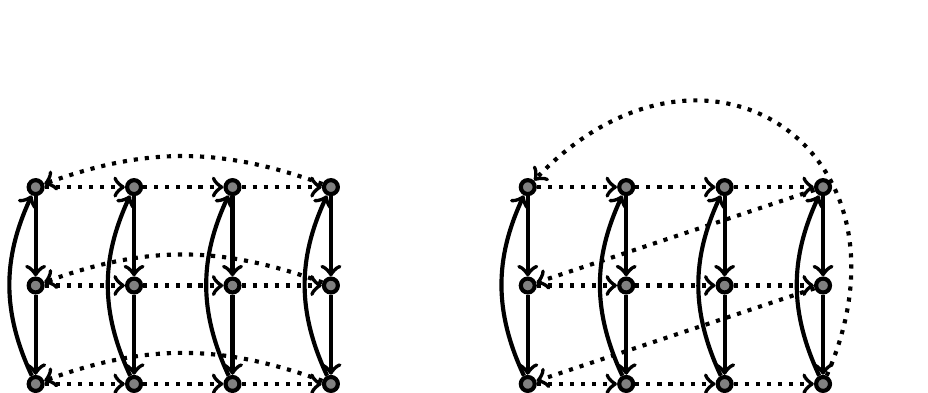}
    \caption{Tori of height $3$ and width $4$ with twist $0$ (left)
      and twist $1$ (right). The $E_1$- and $E_2$-edges are depicted by solid
      arcs and dotted arcs, respectively.  }
    \label{figure:torus-and-twistedtorus}
  \end{figure}
        
  We consider disjoint unions of tori of height $h$. The \emph{twist}
  of a disjoint union $\A$ of tori $\T_{1}, \ldots, \T_{\ell}$ is
  defined as $\twist(\A) := \twist(\T_{1}) + \dotsb +
  \twist(\T_{\ell}) \mymod h$. This example presents an
  $\ordinvFOMODnum{h}(\sigma)$-sentence $\varphi_{h}$ which defines
  the class of disjoint unions of tori of height $h$ and twist $0$ on
  the class of finite $\sigma$-structures.

  Using the Hanf-locality of $\FOMOD{h}$ (cf. \cite{HellaLibkinNurmonen1999}), it can be
  shown that plain $\FOMOD{h}$ cannot distinguish between tori with
  twist $0$ and tori with twist $1$ if their width is sufficiently
  large, depending on the sentence. Here, we show that
  $\ordinvFOMODnum{h}$ can distinguish between tori with twist $0$ and
  tori with twist $1$.  To this end, we use the order to choose a
  distinguished element from each column of a torus.
  \parpic[r]{\includegraphics{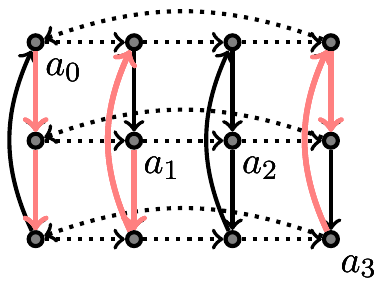}}

  A \emph{list of representatives} for a torus
  $\T$ of width $w$ and height $h$ is a tuple $R=(a_0, \ldots, a_{w-1})$ of nodes of
  $\T$ where, for each $j\in [w]$, \ $a_{j}$ belongs to the $j$-th
  column of $\T$.  The \emph{$j$-th turn path} with respect to $R$ is
  the directed $E_1^{\T}$-path from $a_{j}$ to the
  $E_2^{\T}$-successor of $a_{(j-1 \mymod w)}$. The length of the $j$-turn path w.r.t $R$ is called the \emph{$R$-turn of the $j$-th column of $\T$},
  in symbols: $\tturn_{j}(\T, R)$. Note that $\tturn_{j}(\T, R) \in [h]$.
  The \emph{$R$-turn of $\T$} is defined as $\tturn(\T,R) :=
  \tturn_{0}(\T,R) + \dotsb + \tturn_{w-1}(\T,R)$.  See the picture to
  the right for an illustration of a torus with a set of
  representatives $R$ with $R$-turn $0$; the turn paths are drawn thick and red.
  \begin{claim}
    $\tturn(\T,R) \equiv \twist(\T) \mymod h$, for every list $R$ of
    representatives for $\T$.
  \end{claim}
  \begin{proof}
    Let $k := \twist(\T)$.  First note that if
    $R=(a_0,\ldots,a_{w-1})$ is such that $a_j=(0,j)$, for every $j\in
    [w]$, then the following is true: $\tturn_{j}(\T, R)=0$ for every
    $j\in\set{1,\ldots,w{-}1}$ (as $a_j$ is the $E_2^\T$-successor of
    $a_{j-1}$), and $\tturn_{0}(\T,R)=k$ (as $a_0=(0,0)$ and the
    $E_2^\T$-successor of $(0,w{-}1)$ is $(k,0)$). Thus, $\tturn(\T,R)
    = k$.

    Now consider arbitrary lists $R=(a_0,\ldots,a_{w-1})$ of
    representatives.  By induction on the number $n_R$ of elements
    $a_{j}$ such that $a_{j}\neq (0,j)$, we show that \ $\tturn(\T,R)
    \equiv k \mymod h$.

    The induction base where $n_R=0$ has been shown already.  For the
    induction step, let $j$ be such that $a_j= (p,j)$ for some $p \in
    \set{1, \ldots, h{-}1}$.  Let $R'$ be obtained from $R$ by
    replacing $a_j$ with the element $(0,j)$.  By induction we know
    that $\tturn(\T,R')\equiv k\mymod h$. By a straightforward case
    distinction, one can verify that the following is true:
                            \ $\tturn_{j}(\T,R') \equiv \tturn_{j}(\T,R)+p \mymod h$ \ and
    \ $\tturn_{j{+}1\mymod w}(\T,R') \equiv \tturn_{j{+}1 \mymod
      w}(\T,R) -p \mymod h$. Furthermore,
    $\tturn_{j'}(\T,R')=\tturn_{j'}(\T,R)$ for all $j' \in [w]
    \setminus \set{\,j,\ j{+}1\mymod w\,}$.  Hence, $\tturn(\T,R)
    \equiv \tturn(\T,R') \mymod h$. This completes the proof of the
    claim.
  \end{proof}
  Now consider a disjoint union $\A$ of tori $\T_{1}, \ldots,
  \T_{\ell}$ of height $h$. From the previous claim, we obtain
  \begin{eqnarray*}
    \twist(\A) &\equiv &\twist(\T_{1}) + \dotsb + \twist(\T_{\ell}) \mymod h \\ &\equiv &\tturn(\T_{1},R_{1}) + \dotsb + \tturn(\T_{\ell},R_{\ell}) \mymod h,
  \end{eqnarray*}
  for all lists of representatives $R_{1}, \ldots, R_{\ell}$ of $\T_{1}, \ldots, \T_{\ell}$.
  
  We proceed with the construction of an $\FOMOD{h}(\sigma \union
  \set{<})$-sentence $\psi_{h}$ which, when evaluated in a disjoint
  union of tori $\T_{1}, \ldots, \T_{\ell}$ of height $h$ and width
  $w_{1}, \ldots, w_{\ell}$, computes $\tturn(\T_{1},R_{1}) + \dotsb +
  \tturn(\T_{\ell},R_{\ell}) \mymod h$, for lists of representatives
  $R_{1}, \ldots, R_{\ell}$ of $\T_{1}, \ldots, \T_{\ell}$ which
  depend on $<$. The sentence $\psi_{h}$ is satisfied iff this sum
  modulo $h$ is $0$, i.e., iff $\twist(\A) = 0$.  The sentence uses the
  order $<$ only to choose a particular list of representatives
  $R=(a_0,\ldots,a_{w-1})$ for each torus $\T$ of the disjoint union
  $\A$ by letting $a_j$ be the smallest element of the $j$-th
  column. Hence, $\psi_{h}$ is order-invariant on the class of all
  disjoint unions of tori of height $h$. To compute $\twist(\A)$,
  the formula uses a formula $\psi_{\textit{on-turn-path}}(x)$ which
  is satisfied by exactly those elements $b$ in
  the $j$-th column, for each  $j\in [w]$, of a torus $\T$ that are different from $a_j$ and
  which lie on the $j$-turn path w.r.t $R$. Then, $\twist(\T) = \tturn(\T,R)$ is exactly the number,
  modulo $h$, of all nodes $b$ of $\T$ for which
  $\T\models\psi_{\textit{on-turn-path}}[b]$. Consequently,
  $\twist(\A)$ is exactly the number, modulo $h$, of all nodes $b$ of
  $\A$ for which $\A\models\psi_{\textit{on-turn-path}}[b]$.  Thus, we
  can choose
  \[
  \psi_h \ := \ \ \existsc{h} x\ \psi_{\textit{on-turn-path}}(x).
        \]

  It remains to define the formula $\psi_{\textit{on-turn-path}}(x)$.
  Since the height $h$ is fixed, it is easy to see that there are
  $\FO(\sigma)$-formulas $\varphi_{\textit{col}}(x,y)$ and
  $\varphi_{\textit{between}}(x,y,z)$, and an $\FO(\sigma \union
  \set{<})$-formula $\varphi_{R}(x)$, such that for each torus $\T$
  and all elements $a,b,c\in T$ the following holds:
  \begin{itemize}
  \item $\T\models\varphi_{\textit{col}}[a,b]$ iff $a$ and $b$ belong
    to the same column of $\T$.
  \item $\T\models\varphi_{\textit{between}}[a,b,c]$ iff $a,b,c$
    belong to the same column, $b\neq a$, and $b$ lies on the directed
    $E_1^{\T}$-path of length at most $h{-}1$ from $a$ to $c$.
  \item $\T\models\varphi_{R}[a]$ iff $a$ is the smallest element with
    respect to $<$ in its column.
  \end{itemize}

\noindent
Now, we can choose
\begin{align*}
  \psi_{\textit{on-turn-path}}(x) \ := \quad \exists y\, \exists z\,
  \exists z' \; \big(\ \varphi_{R}(y)\ \land\ \varphi_{R}(z)\ \land\
  E_2(z,z')\ \land\ \varphi_{\textit{between}}(y,x,z') \ \big).
\end{align*}

\medskip

Finally, we want to construct the sentence $\varphi_{h}$ which is
order-invariant on the class of all finite $\sigma$-structures and
which defines the class of disjoint unions of tori with twist $0$.  It
is not difficult to see that a finite $\sigma$-structure $\A$ is
isomorphic to a disjoint union of tori iff it satisfies the following
three properties:
\begin{itemize}
\item The relations $E_{1}^{\A}$ and $E_{2}^{\A}$ are graphs of permutations $\pi_{1}$ and $\pi_{2}$ of $A$
  (i.e. disjoint unions of cycles).
\item The cycles of $E_{1}^{\A}$, which we call \emph{columns},
  all have length $h$.
\item The cycles of $E_{2}^{\A}$ all have length at least $2$.
\item The $\pi_{2}$-image of each column $\C$ is a column $\C'$ and the
  restriction of $\pi_{2}$ to $\C$ is an isomorphism of $\C$ and $\C'$.
\end{itemize}
Note that these properties can be expressed by an
$\FO(\sigma)$-sentence $\theta_h$. For the last property, this can be
done easily since the length $h$ of the columns is fixed. We let
$\varphi_{h} := \theta_{h} \land \psi_{h}$.  Since $\psi_{h}$ is
order-invariant on disjoint unions of tori and $\theta_{h}$ defines the
class of all disjoint unions of tori, $\varphi_{h}$ is
order-invariant on the class of all finite $\sigma$-structures.

\end{example}

\section{Locality of queries}\label{section:gaifman}

A \emph{$k$-ary query} $q$ is a mapping that associates with every
finite $\sigma$-structure $\A$ a $k$-ary relation $q(\A)\subseteq
A^k$, which is invariant under isomorphisms, i.e., if $\pi$ is an
isomorphism from a $\sigma$-structure $\A$ to a $\sigma$-structure $\B$, then for all
$\ov{a}=(a_1,\ldots,a_k)\in A^k$ we have $\ov{a}\in q(\A)$ iff
$\pi(\ov{a}) = (\pi(a_1),\ldots,\pi(a_k)) \in q(\B)$.
If $\Class$ is a class of finite $\sigma$-structures, then
every $\arbinvFOMODpC(\sigma)$-formula $\varphi(\ov{x})$ with $k$
free variables defines a $k$-ary query $q_\varphi$ on $\Class$ via
$q_\varphi(\A)=\setc{\ov{a}\in A^k}{\A\models \varphi[\ov{a}]}$, for
every $\sigma$-structure $\A\in\Class$. 

The \emph{Gaifman graph} of a $\sigma$-structure $\A$ is the undirected
graph $\G(\A)$ with vertex set $A$, where for any $a,b\in
A$ with $a\neq b$ there is an undirected edge between $a$ and $b$ iff
there is an $R\in \sigma$ and a tuple $(a_1,\ldots,a_{\ar(R)})\in R^\A$ such that
$a,b\in\set{a_1,\ldots,a_{\ar(R)}}$.
The \emph{distance} $\dist^\A(a,b)$ between two elements $a,b\in A$ is the length
of a shortest path between $a$ and $b$ in $\G(\A)$;
in case that there is no path between $a$ and $b$ in $\G(\A)$, we let $\dist^\A(a,b) := \infty$.
The distance $\dist^\A(b,\ov{a})$ between an element $b\in A$ and a
tuple $\ov{a}=(a_1,\ldots,a_k)\in A^k$ is the the minimum of
$\dist^\A(b,a_i)$ for all $i\in\set{1,\ldots,k}$.
For every $r\in\NN$, the \emph{$r$-ball} $N_r^\A(\ov{a})$ around a
tuple $\ov{a}\in A^k$
is the set of all elements $b$ with $\dist^\A(b,\ov{a})\leq r$.
The \emph{$r$-neighbourhood} of $\ov{a}$ is the induced substructure $\NS_r^\A(\ov{a})$
of $\A$ on $N_r^\A(\ov{a})$.

\subsection{Gaifman locality}\label{subsection:GaifmanLocality}

The notion of \emph{Gaifman locality} provides a standard tool for showing
that particular queries are not definable in certain logics (cf.,
e.g., the textbook \cite{Libkin2004} for an overview).
\begin{definition}[Gaifman locality]\label{def:gaifman-locality}
Let 
$\Class$ be a class of finite $\sigma$-structures,
$k\in\NNpos$ and $f:\NN\to\NN$.
A $k$-ary query $q$ is \emph{Gaifman $f(n)$-local
on $\Class$} 
if there is an $n_0\in\NN$ such that for every $n\in\NN$ with $n\geq n_0$
and every $\sigma$-structure 
$\A\in\Class$
with $|A|=n$, the following is true
for all $k$-tuples $\ov{a},\ov{b}\in A^k$ with 
$(\NS_{f(n)}^\A(\ov{a}),\ov{a}) \isom (\NS_{f(n)}^\A(\ov{b}),\ov{b})$: \quad
$\ov{a}\in q(\A) \iff \ov{b}\in q(\A)$.
\\
The query $q$ is \emph{Gaifman $f(n)$-local} if it is Gaifman
$f(n)$-local on the class of all finite $\sigma$-structures.
\end{definition}
That is, in a $\sigma$-structure of cardinality $n$, a query that is Gaifman $f(n)$-local
cannot distinguish between $k$-tuples of nodes whose neighbourhoods of radius $f(n)$ are isomorphic.
The function $f(n)$ is
called the \emph{locality radius} of the query.
As a simple example, consider the unary query
which maps an undirected graph to the set of all vertices
which are the central vertex of a star with an even number of rays. Here, a \emph{star} is a connected graph where all but the unique central vertex are rays, i.e.
vertices of degree one. This query is Gaifman $2$-local on the class of all undirected graphs.
As an example of a query which is not Gaifman $f(n)$-local for any sublinear function $f$, consider the binary reachability query which maps a graph to all pairs of vertices such that there exists a path from the first vertex to the second vertex.

It is well-known that queries definable in $\FO$ or $\FOMODp$ (for
any $p\geq 2$) are Gaifman local with a constant locality
radius \cite{HellaLibkinNurmonen1999}. The articles \cite{Grohe1998} and \cite{AMSS-SICOMP}
generalised this to order-invariant $\FO$ (for constant locality
radius) and arb-invariant $\FO$ (for polylogarithmic locality radius) in
the following sense:
Let $k\in\NNpos$, and let $q$ be a $k$-ary query.
If $q$ is definable in $\ordinvFO(\sigma)$, then there is a $c\in\NN$ such
that $q$ is Gaifman $c$-local.
If $q$ is definable in
$\arbinvFO(\sigma)$, then there is a  $c\in\NN$ such that $q$ is
Gaifman $(\log n)^c$-local. 
However, for every $d\in\NN$ there is a unary query $q_d$ that is
definable in $\arbinvFO(\set{E})$ and that is not Gaifman $(\log n)^d$-local.

Somewhat surprisingly, 
using Example~\ref{example:niemistoe-torus} one obtains that the
Gaifman locality result
cannot be generalised to order- or arb-invariant $\FOMODp$.
In fact, $\ordinvFOMODp$ can define queries that are 
not even Gaifman local with locality radius as big as
$(\frac{n}{h}{-}2)$, for the smallest prime divisor $h$ of $p$:
\begin{proposition}\label{prop:notGaifmanLocal}
 Let $h\in\NN$ with $h\geq 2$, and let $\sigma=\set{R,E_1,E_2}$ be a signature consisting of a unary 
 relation symbol $R$ and two binary relation symbols $E_1,E_2$. 
 There exists a unary query $q$ that is
 not Gaifman $(\frac{n}{h}{-}2)$-local, but
 definable in $\ordinvFOMODnum{p}(\sigma)$, for every multiple $p\geq 2$ of $h$.
\end{proposition} 
\begin{proof}
 Recall the $\set{E_1,E_2}$-structures from Example~\ref{example:niemistoe-torus} which we called tori.
 From a torus $\T$ of height $h$ and width $w$, we obtain a $\sigma$-structure
 $\HS$ by deleting all $E_{2}$-edges from the last to the first column
 and marking the least element of the last column with a unary relation.
 We call the structure $\HS$ a \emph{hose}.
 More precisely, the universe of $\HS$ is $H:=T=[h]\times [w]$
 and the relations of $\HS$ are
 \begin{eqnarray*}
   &E_{1}^{\HS} &:= \ E_{1}^{\T},\\
   &E_{2}^{\HS} &:= \ E_{2}^{\T} \setminus
\setc{\ \big((i,w{-}1)\,,\,(i{+}\twist(\T) \mymod h,\, 0)\big)\ }{\;i\in [h]\;}, \\
&R^{\HS} &:= \ \set{\, (0,w{-}1)\, }.
\end{eqnarray*}

\begin{figure}[t]    
  \centering\vspace{-6ex}
  \includegraphics{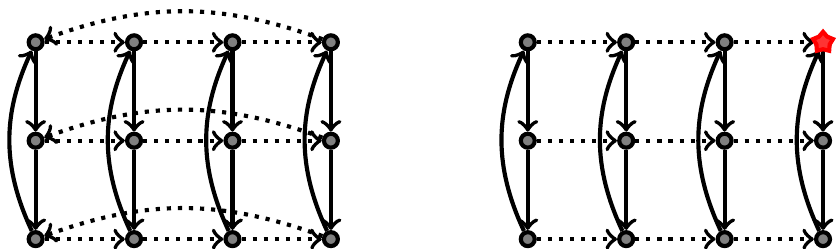}
  \caption{A torus $\T$ (left) and the hose $\HS$ (right) obtained from $\T$. The unique node in the relation $R^{\HS}$ of the hose is the rightmost node in the top row, depicted by a red star.}
\label{figure:torus-and-hose}
\end{figure}

 From Example~\ref{example:niemistoe-torus} we obtain an $\ordinvFOMODnum{h}(\set{E_1,E_2,<})$-sentence
 $\varphi_{h\textit{-torus}}$ which defines the class of all tori $\T$ of height $h$ and $\twist(\T)=0$
 on the class of all finite $\set{E_{1},E_{2}}$-structures.
 We modify $\varphi_{h\textit{-torus}}$ in such a way that we obtain
 an $\ordinvFOMODnum{h}(\sigma \cup \set{<})$-formula 
 $\psi(x)$ which, when evaluated in a hose $\HS$ with $x$ interpreted as the element $a:=(0,0)$
 or $b:=(1,0)$, 
 simulates $\varphi_{h\textit{-torus}}$ evaluated on a torus of height $h$ with twist $0$ or twist $1$, respectively.
 To this end, we let $\psi(x)$ state that each of the following is satisfied:
 \begin{itemize}
  \item 
    There is a unique element $y_0$ satisfying $R(y_0)$,
  \item
    there are elements $y_1,\ldots,y_{h-1}$ such that $E_1(y_i,y_{i+1
      \mymod h})$ is true
    for all  $i\in [h]$, 
  \item
    there are elements $x_0,\ldots,x_{h-1}$ such that $x_0{=}x$ and
    $E_1(x_i,x_{i+1 \mymod h})$ is
    true for all $i\in [h]$, 
  \item 
    the formula $\varphi'$ is satisfied, where $\varphi'$ is obtained from 
    $\varphi_{h\textit{-torus}}$ by replacing every atom of 
    the form $E_2(u,v)$ by the formula
    $
      \big( \, E_2(u,v) \ \oder \ \Oder_{0\leq i<h} \big( u{=}y_i \ \und \ v{=}x_i \big) \, \big).
    $
 \end{itemize}

 \noindent
 Clearly, $\HS\models \psi[a]$ (since
 $\T\models\varphi_{h\text{-torus}}$ if $\twist(\T)=0$), 
 and $\HS\not\models\psi[b]$ (since $\T\not\models\varphi_{h\text{-torus}}$ if $\twist(\T) = 1$). 
 Thus, $a\in q_\psi(\HS)$ and $b\not\in q_\psi(\HS)$.
 Note that the $(w{-}2)$-neighbourhoods of $a$ and $b$ in the
 hose $\HS$ are isomorphic, i.e.,  
 $(\NS^{\HS}_{w-2}(a),a)\isom (\NS^{\HS}_{w-2}(b),b)$. 
 The cardinality of $\HS$ is $n:=h w$, and hence $w{-}2=\frac{n}{h}{-}2$.
 Thus, the query defined by $\psi(x)$ is not
 Gaifman $(\frac{n}{h}{-}2)$-local.

 By Example~\ref{example:niemistoe-torus}, $\varphi_{h\textit{-torus}}$ is order-invariant on the 
 class of all finite $\set{E_1,E_2}$-structures. Therefore,
 the formula $\psi(x)$ is order-invariant on the class of all finite 
 $\sigma$-structures.
 Furthermore, $\varphi_{h\textit{-torus}}$ uses only modulo counting quantifiers with modulus $h$
 and the construction of $\psi(x)$ adds no new modulo counting quantifiers.
 As we observed in Section~\ref{section:preliminaries}, if $p$ is a multiple of $h$, each modulo counting quantifier with modulus $h$
 can also expressed using quantifiers with modulus $p$.
\end{proof}

\subsection{Weak Gaifman locality}\label{subsection:WeakGaifmanLocality}

\emph{Weak Gaifman locality} (cf., \cite{Libkin2004}) is a relaxed notion of Gaifman locality where ``$\ov{a}\in q(\A)\iff
\ov{b}\in q(\A)$'' needs to be true only for those tuples $\ov{a}$ and $\ov{b}$ whose $f(n)$-neighbourhoods
 are disjoint.

\begin{definition}[Weak Gaifman locality]\label{def:weak-gaifman-locality}
Let 
$\Class$ be a class of finite $\sigma$-structures,
$k\in\NNpos$ and $f:\NN\to\NN$.
A $k$-ary query $q$ is \emph{weakly Gaifman $f(n)$-local on $\Class$} if
there is an $n_0\in\NN$ such that for every $n\in\NN$ with $n\geq n_0$
and every $\sigma$-structure $\A\in\Class$ with $|A|=n$, the following is true
for all $k$-tuples $\ov{a},\ov{b}\in A^k$ with 
\ $(\NS_{f(n)}^\A(\ov{a}),\ov{a}) \isom (\NS_{f(n)}^\A(\ov{b}),\ov{b})$
\ and \
$N_{f(n)}^\A(\ov{a})\cap 
N_{f(n)}^\A(\ov{b})=\emptyset$: \quad
$\ov{a}\in q(\A) \iff \ov{b}\in q(\A)$.
\ The query $q$ is \emph{weakly Gaifman $f(n)$-local} if it is
weakly Gaifman $f(n)$-local on the class of all finite $\sigma$-structures.
\end{definition}

Note that the example presented in the
proof of Proposition~\ref{prop:notGaifmanLocal} does not provide a
counter-example to \emph{weak} Gaifman locality, since the elements
$a$ and $b$ considered in the proof of
Proposition~\ref{prop:notGaifmanLocal} are of distance 1, and thus
their $f(n)$-neighbourhoods are not disjoint. 
However, using Example~\ref{example:niemistoe-evencycles}, one
obtains a counter-example to \emph{weak} Gaifman locality for
$\ordinvFOMODnum{p}$ for \emph{even} numbers $p$; see Proposition 6.23 in \cite{Niemistoe-PhD}. 
Here, we present a refinement of Niemist\"o's proof which
provides a counter-example to weak Gaifman locality already
for the restricted case of string structures.

\begin{proposition}\label{prop:notWeaklyGaifmanLocal}
 Let $\Sigma:=\set{0,1}$, and let
 $\sigma_\Sigma=\set{E,P_0,P_1}$ be the signature used for representing strings
 over $\Sigma$.
 There exists a unary query $q$ that is
 not weakly Gaifman $(\frac{n}{4}{-}1)$-local on $\ClassStrings$, but
 definable in $\ordinvFOMODnum{p}(\sigma_\Sigma)$, for every even number
 $p\geq 2$.
\end{proposition}
\begin{proof}
For every $\ell\in\NNpos$, let $\A_{\ell}$ and $\B_{\ell}$ be $\set{E}$-structures
whose universe consists of $2\ell$ vertices,
the edge relation of $\A_\ell$ consists of two directed cycles 
of length $\ell$, and the
edge relation of $\B_\ell$ consists of a single directed cycle
of length $2\ell$.
Furthermore, we choose $w_\ell$ to be the string
\ $ 1^{\ell}\,0^{\ell}\,1^{\ell}\,0^{\ell}$, \
and we let $a_\ell:=\ell$ be the rightmost position of the first block of
$1$s, and $b_\ell:=3\ell$ the rightmost position of the second
block of $1$s.

From Example~\ref{example:niemistoe-evencycles} we obtain an
$\ordinvFOMODnum{2}(\set{E})$-sentence
$\varphi_{\textit{even cycles}}$ that is satisfied by a finite
$\set{E}$-structure $\A$ iff $\A$ is a disjoint union of directed
cycles where the number of cycles of even length is even. Thus, for
every $\ell\in\NNpos$ we have: \ $\A_\ell\models\varphi_{\textit{even
    cycles}}$ and $\B_\ell\not\models\varphi_{\textit{even cycles}}$.
We modify the formula $\varphi_{\textit{even cycles}}$ in such a way
that we obtain an $\ordinvFOMODnum{2}(\sigma_\Sigma)$-formula $\psi(x)$
which, when evaluated in 
the $\sigma_\Sigma$-structure $\S_{w_\ell}$ representing the string $w_\ell$
with $x$ interpreted as
the position $a_\ell$ or $b_\ell$ simulates $\varphi_{\textit{even cycles}}$ evaluated on $\A_\ell$ or $\B_\ell$, respectively.
To this end, we let $\psi(x)$ be a formula stating that each of the
following is satisfied:
\begin{figure}[t]
  \centering
  \includegraphics{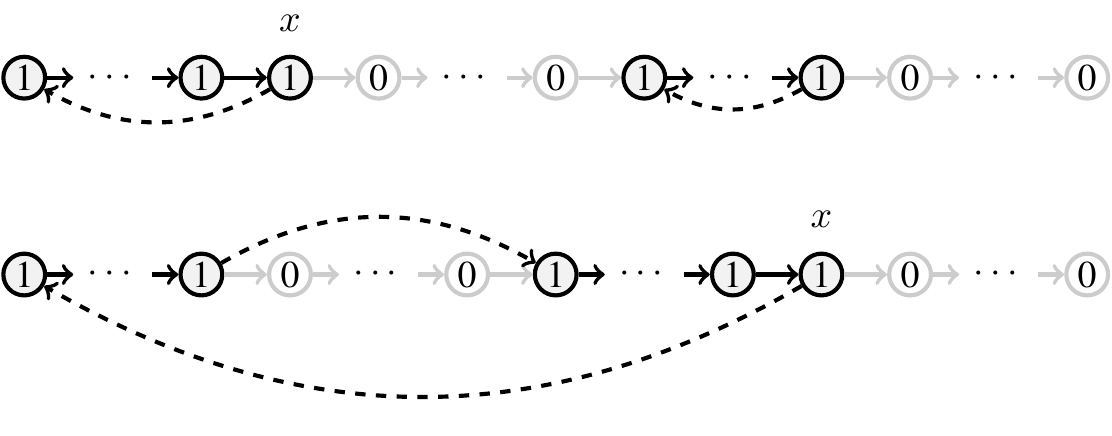}
  \caption{A string structure $\S_{w_{\ell}}$ with $x=a_{\ell}$ (top) and $x=b_{\ell}$ (bottom). The formula $\psi(x)$ simulates $\varphi_{\textit{even cycles}}$ on the structure where the dashed edges are added and the light nodes and edges are removed. }
\end{figure}

\begin{itemize}
\item 
  There is a unique position $x'\neq x$ that carries the letter 1 such that the
  position directly to the right of $x'$ carries the letter 0.
 \item 
  There is a unique position $y$ of in-degree $0$, and this position
  carries the letter 1.
  Furthermore, there is a unique position $y'$ that carries the letter 1, such
  that the position directly to the left of $y'$ carries the letter 0.
 \item The formula $\varphi'$ is satisfied, where $\varphi'$ is
   obtained from $\varphi_{\textit{even cycles}}$ by
   relativisation of all quantifiers to positions
   that carry the letter $1$, and by
   replacing every
   atom of the form $E(u,v)$ by the formula 
   \ $\big( E(u,v) \oder 
      \big( u{=}x  \und v{=}y  \big) \oder 
      \big( u{=}x' \und v{=}y' \big)\big)$.
\end{itemize}

\noindent
Clearly, for every $\ell\in\NNpos$ we have: \ $\S_{w_\ell}\models\psi[a_\ell]$ 
(since $\A_\ell\models\varphi_{\textit{even cycles}}$), and
$\S_{w_\ell}\not\models\psi[b_\ell]$ (since
$\B_\ell\not\models\varphi_{\textit{even cycles}}$).
Thus,
$a_\ell\in q_{\psi}(\S_{w_{\ell}})$ and $b_\ell\not\in q_{\psi}(\S_{w_{\ell}})$.
Note that the $(\ell{-}1)$-neighbourhoods of $a_\ell$ and $b_\ell$ in $\S_{w_{\ell}}$ are
disjoint and isomorphic. The cardinality of $\S_{w_{\ell}}$ is $n:=4\ell$, and
hence $\ell{-}1=\frac{n}{4}{-}1$. Thus, the query
defined by $\psi(x)$ is not weakly Gaifman $(\frac{n}{4}{-}1)$-local.
Since $\varphi_{\textit{even cycles}}$ is order-invariant on all
finite $\set{E}$-structures, 
the formula $\psi(x)$ is order-invariant on the class of all finite
$\sigma_\Sigma$-structures.
Note that $\psi(x)$ is definable in $\ordinvFOMODnum{2}(\sigma_\Sigma)$, and
hence also in $\ordinvFOMODnum{p}(\sigma_\Sigma)$, for every multiple $p$ of
$2$.
\end{proof}

In light of Proposition~\ref{prop:notWeaklyGaifmanLocal} it is 
somewhat surprising that for \emph{odd} numbers $p$, unary queries
definable in $\ordinvFOMODp$
are weakly Gaifman local with constant locality radius --- this is a result 
obtained by Niemist\"o
(see Corollary~6.37 in \cite{Niemistoe-PhD}).
For \emph{odd prime powers $p$} we can generalise this to $k$-ary
queries definable in 
$\arbinvFOMODp$, when allowing polylogarithmic locality radius.
Note that we cannot hope for a smaller locality radius, since
\cite{AMSS-SICOMP} provides, for every $d\in\NN$, a unary query
definable in $\arbinvFO(\set{E})$ that is not weakly Gaifman $(\log n)^d$-local.
Precisely, we will show the following.
\begin{theorem}\label{thm:weakGaifmanLocality}
Let $\Class$ be a class of finite $\sigma$-structures.
Let $k\in\NNpos$, let $q$ be a $k$-ary query, and let $p$ be an odd
prime power. If $q$ is definable
in $\arbinvFOMODpC(\sigma)$ on $\Class$, then there is a $c\in\NN$ such that $q$ is
weakly Gaifman $(\log n)^c$-local on $\Class$.
\end{theorem}

The proof of this theorem will be given in the next
subsection, as an easy consequence of
Theorem~\ref{thm:CyclicShiftLocality} below.
A generalisation of Theorem~\ref{thm:weakGaifmanLocality} from odd
prime powers to arbitrary odd numbers $p$ would lead to new
separations concerning circuit complexity classes and can therefore be
expected to be rather difficult (see Remark~\ref{remark:circuitcomplexity}).

\subsection{Shift locality}\label{section:ShiftLocality}
\mbox{}\\
The following notion of \emph{shift locality} 
generalises the notion of \emph{alternating Gaifman locality}
introduced by Niemist\"o in \cite{Niemistoe-PhD}.
In some sense discussed below, it unifies this notion
and the notion of weak Gaifman locality.
\begin{definition}[Shift locality]\label{def:CyclicShiftLocality}
Let $\Class$ be a class of finite $\sigma$-structures.
Let $k,t\in\NNpos$ with $t\geq 2$, and let $f:\NN\to\NN$. 
A $kt$-ary query $q$ is \emph{shift $f(n)$-local w.r.t.\ $t$ on $\Class$} if
there is an $n_0\in\NN$ such that for every $n\in\NN$ with $n\geq
n_0$ and every $\sigma$-structure $\A\in\Class$ with $|A|=n$, the following is
true for all $k$-tuples $\ov{a}_0,\ldots,\ov{a}_{t-1}\in A^k$ with
$(\NS^\A_{f(n)}(\ov{a}_i),\ov{a}_i)\isom
(\NS^\A_{f(n)}(\ov{a}_j),\ov{a}_j)$ and
$N^\A_{f(n)}(\ov{a}_i)\cap N^\A_{f(n)}(\ov{a}_j)=\emptyset$ for all
$i,j\in [t]$ with $i\neq j$: \quad 
$(\ov{a}_0,\ov{a}_1\ldots,\ov{a}_{t-1})\in q(\A) \iff
(\ov{a}_1,\ldots,\ov{a}_{t-1},\ov{a}_0)\in q(\A)$.
\\
Query $q$ is \emph{shift $f(n)$-local w.r.t.\ $t$} if it is
shift $f(n)$-local w.r.t.\ $t$ on the class of all finite $\sigma$-structures. 
\end{definition}
The case of $k=1$ and $t=3$ for a constant function $f$ yields Niemist\"o's notion
of alternating Gaifman locality \cite{Niemistoe-PhD}.
From query $q$ of Proposition~\ref{prop:notWeaklyGaifmanLocal}, we obtain an example of an alternatingly Gaifman local query.
For all even numbers $p\geq 2$, this query $q$ is definable by an $\ordinvFOMODnum{p}$-formula $\psi(x)$.
The query defined by the formula $\phi(x,y,z) := \psi(x) \land y=y \land z=z$ is an alternatingly Gaifman local query. 
This can be seen directly, but it is also
a consequence of the result of \cite{Niemistoe-PhD} that 
queries which are definable by formulas of $\ordinvFOMODnum{p}$ for even numbers $p$
are alternatingly Gaifman local. Several examples of non-shift locality
will be given in Section~\ref{subsection:Application}.
To understand the relation between shift locality and weak Gaifman locality,
consider a $k$-ary query $q$ and the $2k$-ary query $\tilde q$ with $\tilde q(\A) := \setc{\ov a_{1}\ov a_{2}}{\ov a_{1} \in q(\A), \ov a_{2} \in A^{k}}$. Then $q$ is weakly Gaifman $f(n)$-local iff $\tilde q$ is shift $f(n)$-local w.r.t. $2$.
The notion of shift locality helps to discuss both kinds of locality in a uniform way.

In a technical lemma (Lemma~6.36 in \cite{Niemistoe-PhD}), Niemist\"o
showed that for $k=1$ and $p,t\in\NN$ with $p,t\geq 2$ and $p$ and $t$ coprime, for
every $t$-ary query $q$ definable in $\ordinvFOMODp(\sigma)$, there is
a $c\in\NN$ such that $q$ is shift $c$-local w.r.t.\ $t$.
Our next result deals with the general case of shift locality
and the more expressive logic $\arbinvFOMODp(\sigma)$, when allowing
polylogarithmic locality radius. 

\begin{theorem}\label{thm:CyclicShiftLocality}
Let $\Class$ be a class of finite $\sigma$-structures.
Let $k,t\in\NNpos$ with $t\geq 2$, let $q$ be a $kt$-ary query, and
let $p$ be a prime power such that $p$ and $t$ are coprime. 
If $q$ is definable in $\arbinvFOMODpC(\sigma)$ on $\Class$, then there is a
$c\in\NN$ such that $q$ is shift $(\log n)^c$-local w.r.t.\ $t$ on $\Class$.
\end{theorem}

Our proof of Theorem~\ref{thm:CyclicShiftLocality} relies on lower
bounds achieved in circuit complexity. 
A generalisation of Theorem~\ref{thm:CyclicShiftLocality} from prime
powers to arbitrary
numbers $p\geq 2$ would lead to new separations of circuit complexity classes
and can therefore be expected to be rather difficult (see
Remark~\ref{remark:circuitcomplexity}). 
Before giving the proof of Theorem~\ref{thm:CyclicShiftLocality}, let us first point out that
it immediately implies Theorem~\ref{thm:weakGaifmanLocality}.

\begin{proofc}{Theorem~\ref{thm:weakGaifmanLocality} (using
    Theorem~\ref{thm:CyclicShiftLocality})} \ 
Let $\varphi(\ov{x})$ be an $\arbinvFOMODpC(\sigma)$-formula
with $k$ free variables $\ov{x}=(x_{1},\ldots,x_{k})$, defining a
$k$-ary query $q_\varphi$ on $\Class$. 
Let $\ov{y}=(y_1,\ldots,y_k)$ be $k$ variables different from the
variables in $\ov{x}$. Then,
\ $
  \psi(\ov{x},\ov{y}) := 
  \big(\,\varphi(\ov{x}) \, \und \Und_{1\leq i\leq k} y_i{=}y_i\,)
$ \
is an $\arbinvFOMODpC(\sigma)$-formula that defines a $2k$-ary query
$q_\psi$. By Theorem~\ref{thm:CyclicShiftLocality}, there exists a $c\in\NN$
such that $q_\psi$ is shift $(\log n)^c$-local w.r.t.\ $t:=2$ on $\Class$.
It is straightforward to see that the shift $(\log n)^c$-locality of
$q_\psi$ w.r.t.\ $t=2$
implies that the query $q_\varphi$ is weakly Gaifman $(\log n)^c$-local.
\end{proofc}

The remainder of this subsection is devoted to the proof of
Theorem~\ref{thm:CyclicShiftLocality}. We follow the overall
method of \cite{AMSS-SICOMP} for the case of disjoint neighbourhoods
(see \cite{Schweikardt2013} for an overview)
and make use of the connection between $\arbinvFOMODp$ and
$\MODp$-circuits \cite{BIS90}, along with a circuit lower
bound by Smolensky \cite{Smolensky}.

We assume that the reader is familiar with basic notions and results
in circuit complexity (cf., e.g., the textbook \cite{AroraBarak}).
A $\MODp$-gate returns the value $1$ iff the number of ones at its
inputs is congruent 0 modulo $p$.
We consider Boolean \emph{$\MODp$-circuits} consisting of AND-, OR-, and
$\MODp$-gates of unbounded fan-in, input gates, negated input gates, and
constant gates $\textbf{0}$ and $\textbf{1}$.
More precisely, a
\emph{$\MODp$-circuit with $m$ input bits} is a directed
acyclic graph whose vertices without ingoing edges are called
\emph{input gates} and are labelled with either $\textbf{0}$,
$\textbf{1}$, $w_\nu$, or $\nicht w_\nu$ for $\nu\in\set{1,\ldots,m}$,
whose internal nodes are called \emph{gates} and are labelled either
AND or OR or $\MODp$, and which has a distinguished vertex without outgoing edges called the \emph{output gate}.
A $\MODp$-circuit $C$ with $m$ input bits naturally defines a function
from $\set{0,1}^m$ to $\set{0,1}$. 
For an input string $w\in\set{0,1}^m$ we say that 
$C$ \emph{accepts} $w$ if $C(w)=1$.
Accordingly, $C$ \emph{rejects} $w$ if $C(w)=0$.
The \emph{size} of a circuit is the number of gates it contains, and
the \emph{depth} is the length of the longest path from any of the input gates 
to the output gate. 

Our proof of Theorem~\ref{thm:CyclicShiftLocality} relies on
Smolensky's following circuit lower bound.
\begin{theorem}[Smolensky \cite{Smolensky}
  (see also \cite{Straubing1994})]\label{thm:smolensky}
  Let $p$ be a prime power and let $r$ be a number which has a prime factor which is different
  from the prime factor of $p$. \\
 There exist numbers $\varepsilon,\ell >0$ such that for every
 $d\in\NNpos$ there is an $m_d\in\NNpos$ such that for
 every $m\in\NN$ with $m\geq m_d$ the following is true:
 \
 No $\MODp$-circuit of depth $d$ and size at most
 $2^{\varepsilon\sqrt[\ell d]{m}}$ accepts exactly those bitstrings
 $w\in\set{0,1}^m$ that contain a number of ones congruent $0$ modulo $r$.
\end{theorem}

In the literature, Smolensky's theorem is usually stated only for
primes $p$. Note, however, that (for each fixed $k\in\NNpos$) $\MOD{p^k}$-gates can easily be
simulated by $\MODp$-circuits of constant depth and polynomial size
(cf., \cite{Straubing1994}), and hence Smolensky's theorem also holds for
prime powers $p$, as stated in Theorem~\ref{thm:smolensky}.
It is still open whether an analogous result also holds for
numbers $p$ composed of more than one prime factor (see Chapter~VIII
of \cite{Straubing1994} and Chapter~14.4 of \cite{AroraBarak} for discussions
on this).

To establish the connection between $\MODp$-circuits and
$\arbinvFOMODp(\sigma)$, we need to represent $\sigma$-structures $\A$
and $K$-tuples $\ov{a}\in A^K$ (for $K\in\NN$) by bitstrings. 
This is done in a straightforward way: 
Let $\sigma=\set{R_1,\ldots,R_{|\sigma|}}$ and let $r_i:=\ar(R_i)$ for each $i\leq |\sigma|$.
Consider a finite $\sigma$-structure $\A$ with $|A|=n$. Let $\iota$ be
an embedding of $\A$ into $[n]$.
For each $R_i\in\sigma$ we let $\Rep^\iota(R_i^\A)$ be the bitstring
of length $n^{r_i}$ whose $j$-th bit is $1$ iff the $j$-th smallest
element in $A^{r_i}$ w.r.t.\ the lexicographic order associated with
$<^\iota$ belongs to the relation $R_i^\A$.
Similarly, for each component $a_i$ of a $K$-tuple
$\ov{a}=(a_1,\ldots,a_K)\in A^K$ we let $\Rep^\iota(a_i)$ be the
bitstring of length $n$ whose $j$-th bit is $1$ iff $a_i$ is the
$j$-th smallest element of $A$ w.r.t.\ $<^\iota$. Finally, we let
\[
  \Rep^\iota(\A,\ov{a}) \ := \ \ 
  \Rep^\iota(R_1^\A) \ \cdots \ \Rep^\iota(R_{|\sigma|}^\A) \ \Rep^\iota(a_1)\
  \cdots \ \Rep^\iota(a_K)
\]
be the \emph{binary representation of $(\A,\ov{a})$ w.r.t.\ $\iota$}. 
Note that, independently of $\iota$, the length of the bitstring
$\Rep^\iota(\A,\ov{a})$ is $\lambda^{\sigma}_K(n):=\sum_{i=1}^{|\sigma|} n^{r_i} + Kn$.

The connection between $\FOMODp(\sigma_\ARB)$ and $\MODp$-circuits is
obtained by the following result.

\begin{theorem}[implicit in \cite{BIS90} (see also \cite{Straubing1994})]\label{thm:formula2circuit}
 Let $\sigma$ be a finite relational signature, let $K\in\NN$, and let
 $p\in\NN$ with $p\geq 2$.
 For every $\FOMODp(\sigma_\ARB)$-formula $\varphi(\ov{x})$ with $K$
 free variables there exist numbers $d,s\in\NN$  such
 that for every $n\in\NNpos$ there is a $\MODp$-circuit $C_n$
 with $\lambda^\sigma_K(n)$ input bits, depth $d$, and size $n^s$ such
 that the following is true for all $\sigma$-structures $\A$
 with $|A|=n$, all $\ov{a}\in A^K$, and all embeddings $\iota$ of $\A$
 into $[n]$: \ \
 $C_n$ accepts $\Rep^\iota(\A,\ov{a})$ $\iff$
 $\A^\iota\models \varphi[\ov{a}]$.
\end{theorem}

The proof of Theorem~\ref{thm:CyclicShiftLocality} can be outlined informally as follows.
Assume that, for some prime power $p$ and a number $t$ which is coprime with $p$, there is an $\arbinvFOMODp(\sigma)$-definable $kt$-ary query $q$ which is not shift $(\log n)^{c}$-local w.r.t $t$ for each $c$.
We wish to obtain a contradiction to Smolensky's theorem (Theorem~\ref{thm:smolensky}).
From the defining arb-invariant sentence for $q$, we obtain a $\MODp$-circuit-family $(C_{n})_{n\in \NN}$ which accepts exactly the bitstrings which are representations of structures $(\A,\ov{a})$ such that $\ov{a}\in q(\A)$ (cf. Theorem~\ref{thm:formula2circuit}).
Since $q$ is not shift $(\log n)^{c}$-local w.r.t $t$, for each $c$,
we obtain an infinite family of counter examples to the shift $(\log n)^{c}$-locality of $q$.
That is, for each $n_{0}$ there is a structure $\A$ on $n \geq n_{0}$ elements
which contains tuples $\ov a_{0}, \ldots, \ov a_{t-1}$ of $k$ elements each such that the tuples have pairwise isomorphic and disjoint $(\log n)^{c}$-neighbourhoods, but $(\ov a_{0}, \ldots, \ov a_{t-1}) \in q(\A)$ and $(\ov a_{1}, \ldots, \ov a_{t-1}, \ov a_{0}) \not\in q(\A)$, i.e. $C_{n}$ accepts all representations of $(\A, \ov a_{0}, \ldots, \ov a_{t-1})$ and rejects all representations of $(\A, \ov a_{1}, \ldots, \ov a_{t-1}, \ov a_{0})$.

We transform each $C_{n}$ into a circuit $\tilde C_{m}$  (cf. Lemma~\ref{lemma:CircuitLemma1} below) of exactly the same size
which accepts all bitstrings where the number of ones is $0$ modulo $t$ and rejects all bitstrings where this number is $1$ modulo $t$; furthermore, it does not distinguish between strings with the same number of ones modulo $t$.
To this end, on input of a word $w$, the circuit $\tilde C_{m}$ simulates the circuit $C_{n}$
on the representation of a structure which is obtained from $\A$ as follows.
For each position $i$ of $w$ which is labelled by $1$, the relations
between elements at distance $i$ and $i+1$ from the tuple $(\ov a_{0}, \ldots, \ov a_{t-1})$ 
are changed in such a way that the resulting structure looks like 
$\A$ where the tuple $(\ov a_{0}, \ldots, \ov a_{t-1})$ has been shifted.
See Figure~\ref{fig:shift} for an example where the structure $\A$ is a graph. 
This will be done in a way which ensures that $\tilde C_{m}$ accepts $w$ if
the number of ones is $0$ modulo $t$ --- in this case, the simulated structure looks like $(\A, \ov a_{0}, \ldots, \ov a_{t-1})$ and hence its representation is accepted by $C_{n}$ ---
and that it rejects $w$ if the number of ones is $1$ modulo $t$ --- in this case, the simulated structure looks like $(\A, \ov a_{1}, \ldots, \ov a_{t-1}, \ov a_{0})$ and hence it is rejected by $C_{n}$.

This construction is not yet sufficient for a contradiction to Smolenky's theorem.
For this, we need to show (cf. Lemma~\ref{lemma:CircuitLemma2} below) that the circuit $\tilde C_{m}$ can be transformed into a circuit $\hat C_{m}$ of roughly the same size and depth which accepts \emph{exactly} the bitstrings with $0$ ones modulo $r$, for some factor $r\geq 2$ of $t$. To this end, we show that there exists a factor $r\geq 2$ of $t$ such that, basically, it is possible to determine
if the number of ones in a bitstring $w$ is $0$ modulo $r$ by computing, for each $j\in [t]$, whether $\tilde C_{m}$ accepts $w$
after replacing the first $j$ zeros of $w$ by ones. 
Having achieved all this, if we fix $c$ and $n_{0}$ appropriately in terms of Smolenky's theorem, we obtain the desired contradiction.

\medskip

We proceed with the formal proof of Theorem~\ref{thm:CyclicShiftLocality}. 
To simplify notation, we define
$\vec{a}^{(0)}:= (\ov{a}_0,\ov{a}_1,\ldots,\ov{a}_{t-1})$ and 
$\vec{a}^{(i)}:=
(\ov{a}_i,\ov{a}_{i+1},\ldots,\ov{a}_{t-1},\ov{a}_0,\ov{a}_1,\ldots,\ov{a}_{i-1})$,
for all
$i\in [t]$ with $i\leq 1$.

\begin{lemma}\label{lemma:CircuitLemma1}
Let $m,k,t\in \NNpos$ with $t\geq 2$. Let $\A$ be a finite
$\sigma$-structure with $n:=|A|$. For each $i\in[t]$ let $\ov{a}_i\in
A^k$ such that for all $i,j\in[t]$ with $i\neq j$ we
have
$(\NS^\A_m(\ov{a}_i),\ov{a}_i)\isom(\NS^\A_m(\ov{a}_j),\ov{a}_j)$ \ and \
$N^\A_m(\ov{a}_i)\cap N^\A_m(\ov{a}_j) = \emptyset$.
Let $p\in\NN$ with $p\geq 2$.
Let $C$ be a $\MODp$-circuit with $\lambda^\sigma_{kt}(n)$ input bits such that:
\begin{enumerate}
\item
$C$ accepts $\Rep^{\iota_1}(\A,\vec{a}^{(i)})$ iff it accepts 
$\Rep^{\iota_2}(\A,\vec{a}^{(i)})$, for all embeddings $\iota_1$ and $\iota_2$
of $\A$ and for every $i\in [t]$, and
\item\label{prop:accept0-reject1} 
$C$ accepts $\Rep^{\iota}(\A,\vec{a}^{(0)})$ and
rejects $\Rep^\iota(\A,\vec{a}^{(1)})$, for every embedding $\iota$ of $\A$.
\end{enumerate}
\noindent
There exists a $\MODp$-circuit $\tilde{C}$ with $m$ input bits, such that:
\begin{enumerate}[label=(\emph{\alph*})]
\item\label{prop:depth-size} 
 $\tilde{C}$ has the same depth and size as $C$,
\item\label{prop:acceptance-depends-on-I} 
 for all $w,w'\in\set{0,1}^m$ with $\abs{w}_1\equiv |w'|_1 \mymod t$, $\tilde{C}$ accepts
$w$ iff it accepts $w'$, and
\item\label{prop:accept-0-reject-1} 
 $\tilde{C}$ accepts all $w\in\set{0,1}^m$ with
$\abs{w}_1\equiv 0 \mymod t$ and rejects all $w\in\set{0,1}^m$ with 
$\abs{w}_1\equiv 1 \mymod t$.
\end{enumerate}
\end{lemma}
\begin{proof}
  Let $I \subset [t]$ be the set containing $i\in [t]$ iff $C$ accepts
  $\Rep^{\iota_1}(\A,\vec{a}^{(i)})$ for some (i.e., due to property
  (1) of $C$, \emph{every}) embedding $\iota_1$ of $\A$.  By property
  (2) of $C$, we know that $0\in I$ and $1 \notin I$.

  For the remainder of this proof, fix an embedding $\iota$ of $\A$
  into $[n]$.  Note that $\iota$ is also an embedding of any other
  $\sigma$-structure that has the same universe as $\A$.  For every
  $w\in\set{0,1}^m$, we will define a $\sigma$-structure $\A_w$ with
  the same universe as $\A$, which has the following property for
  every $i\in [t]$:
  \begin{gather}
    \text{If} \ \ \abs{w}_1\equiv i \mymod t, \ \ \text{then} \ \
    (\A_w,\vec{a}^{(0)}) \ \isom \ (\A,
    \vec{a}^{(i)}).\label{prop:i-cong-t}
  \end{gather}
  Note that if $(\A_w,\vec{a}^{(0)}) \isom (\A,\vec{a}^{(i)})$, then
  there is an embedding $\iota_1$ such that
  $\Rep^{\iota_1}(\A,\vec{a}^{(i)})=\Rep^{\iota}(\A_w,\vec{a}^{(0)})$. Hence,
  due to property (1), $C$ accepts $\Rep^{\iota}(\A_w,\vec{a}^{(0)})$
  iff it accepts $\Rep^\iota(\A,\vec{a}^{(i)})$.

  The circuit $\tilde{C}$ will be constructed so that on input
  $w\in\set{0,1}^m$ it does the same as circuit $C$ does on input
  $\Rep^\iota(\A_w,\vec{a}^{(0)})$.  Thus, the following will be true for
  every $w\in\set{0,1}^m$ and the particular number $i\in [t]$ such
  that $\abs{w}_1\equiv i \mymod t$:
  \begin{equation*}\label{eq:tildeCacceptsI}
    \tilde{C} \text{ accepts } w
    \iff
    C \text{ accepts } \Rep^\iota(\A_w,\vec{a}^{(0)})
    \iff
    C \text{ accepts } \Rep^\iota(\A,\vec{a}^{(i)})
    \iff
    i\in I. \ \ 
  \end{equation*}

\noindent
This immediately implies that $\tilde{C}$ satisfies property (b); and
since $0\in I$ and $1 \notin I$, the circuit $\tilde{C}$ also
satisfies property (c).

\medskip

\begin{figure}
  \begin{center}
    \resizebox{400pt}{!}{\includegraphics{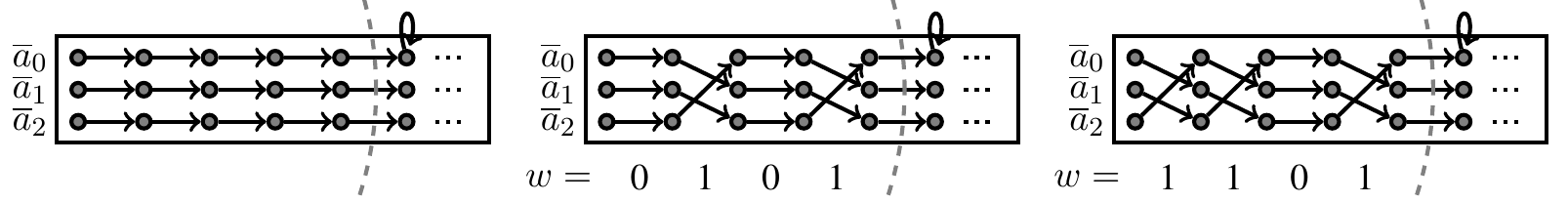}}
    \caption{Illustration of a structure $\A$ (left) and structures
      $\A_w$ for two different bitstrings
      $w$.}\label{fig:shift}
  \end{center}

\end{figure}

\emph{Definition of $\A_w$:} \ For each $j\in [t]$, we partition
$N_m^{\A}(\ov{a}_{j})$ into \emph{shells}
$S_\nu(\ov{a}_{j}):=\setc{x\in A}{\dist^{\A}(x,\ov{a}_{j})=\nu}$, for
all $\nu\in\set{0,\ldots,m}$.  We write $S_\nu$ for the set
$S_\nu(\ov{a}_{0})\union \dotsb \union S_\nu(\ov{a}_{t-1})$.  For each
$j\in[t]$ let $\pi_{j}$ be an isomorphism from
$(\NS^{\A}_{m}(\ov{a}_{j}),\ov{a}_{j})$ to \allowbreak
$(\NS^{\A}_{m}(\ov{a}_{(j+1 \mymod t)}),\ov{a}_{(j+1 \mymod t)})$.
Note that $\pi_{j}(S_\nu(\ov{a}_{j}))=S_\nu(\ov{a}_{(j+1\mymod t)})$
for each $j\in [t]$ and each $\nu\leq m$.

For a bitstring $w=w_1\cdots w_{m}\in\set{0,1}^m$ the structure $\A_w$
has the same universe as $\A$. For each $R\in \sigma$ of arity $r$,
the relation $R^{\A_{w}}$ is obtained from $R^{\A}$ as follows: We
start with $R^{\A_w}:=\emptyset$, and then for each tuple
$\ov{c}\in R^\A$ we insert the tuple $\ov{c}_w$ into $R^{\A_w}$, where
$\ov{c}_w$ is defined as follows:
\begin{enumerate}[label=(\emph{\roman*})]
\item If $\ov{c}\notin (S_{\nu-1} \union S_{\nu})^{r}$ for any
  $\nu \leq m$, or $\ov{c} \in S_{\nu}^{r}$ for some $\nu \leq m$,
  then $\ov{c}_w:=\ov{c}$.
\item Otherwise, if $\ov{c}\in (S_{\nu-1} \union S_{\nu})^{r}$ for
  some $\nu \leq m$, then note that (since $\ov{c}\in R^\A$), there is
  a unique $j\in[t]$ such that
  $\ov{c} \in (S_{\nu-1}(\ov{a}_{j}) \union S_{\nu}(\ov{a}_{j}))^{r}$
  (since $N^\A_m(\ov{a}_{j})\cap N^\A_m(\ov{a}_{j'}) = \emptyset$, for
  all $j,j'\in [t]$ with $j \neq j'$).  To keep the notation simple,
  assume that $\ov{c}=(\ov{c}_{\nu-1}, \ov{c}_{\nu})$, where all
  elements of $\ov{c}_{\nu-1}$ belong to $S_{\nu-1}(\ov{a}_j)$ and all
  elements of $\ov{c}_{\nu}$ belong to $S_{\nu}(\ov{a}_j)$.  We define
  $\ov{c}_w$ depending on the $\nu$-th bit $w_\nu$ of $w$: \ If
  $w_\nu=0$, then $\ov{c}_w:=\ov{c}$. \ If $w_\nu=1$, then
  $\ov{c}_w:= (\ov{c}_{\nu-1}, \pi_{j}(\ov{c}_{\nu}))$.
\end{enumerate}

\noindent
Note that for every $\nu\in\set{1,\ldots,m}$ with $w_\nu=1$, this
construction enforces that the role that was formerly played by shell
$S_\nu(\ov{a}_j)$ is afterwards played by shell
$S_\nu(\ov{a}_{(j+1\mymod t)})$; see Figure~\ref{fig:shift} for an
illustration. This shows that $\A_{w}$ satisfies the property \eqref{prop:i-cong-t}.
A more precise argument will be given at the end of the proof.

\emph{Construction of $\tilde{C}$:}  On input of $w\in\set{0,1}^m$ the circuit $\tilde{C}$ will
simulate circuit $C$ on input $\Rep^\iota(\A_w,\vec{a}^{(0)})$.  We
construct $\tilde{C}$ in a way which mirrors the construction of
$\A_{w}$.  To this end, all inputs of $C$ corresponding, in $\Rep^{\iota}(\A,\vec a^{(0)})$, to tuples in relations of $\A$ that
are unchanged by the construction of $\A_{w}$ (i.e. tuples to which
case (i) of the definition of $\A_{w}$ applies) are fixed to constant
values. Inputs of $C$ corresponding to tuples to which case (ii) of
the definition of $\A_{w}$ applies (i.e. tuples
$\ov{c}\in (S_{\nu-1} \union S_{\nu})^{r}$ for some $\nu \leq m$) are
changed according to the $\nu$-th bit in the input $w$ of $\tilde C$
in a way that mirrors case (ii).

More precisely, for each relation symbol
$R\in \sigma$ of arity $r$ and each tuple $\ov{c}\in R^\A$, we proceed
as follows.  Let $g$ and $\lnot g$ be the non-negated and the negated
input gate of $C$ that correspond to the bit $b$ in
$\Rep^\iota(\A,\vec{a}^{(0)})$ which represents the information
whether $\ov{c}$ belongs to $R^\A$.
\begin{enumerate}[label=(\emph{\roman*})]
\item If $\ov{c}\notin(S_{\nu-1} \union S_{\nu})^{r}$ for all
  $\nu\leq m$, or $\ov{c} \in S_\nu^r$ for some $\nu\leq m$, the gate
  $g$ is replaced by the constant gate \textbf{1}, and the negated
  input gate $\nicht g$ is replaced by the constant gate \textbf{0}.

\item Otherwise, if
  $\ov{c} \in(S_{\nu-1}(\ov{a}_{j})\union S_\nu(\ov{a}_{j}))^{r}$ for
  some $\nu\leq m$ and $j\in [t]$, assume, for simplicity, that
  $\ov{c}:=(\ov{c}_{\nu-1}, \ov{c}_{\nu})$, where all elements of
  $\ov{c}_{\nu-1}$ belong to $S_{\nu-1}(\ov{a}_j)$ and all elements of
  $\ov{c}_\nu$ belong to $S_\nu(\ov{a}_j)$.  Let $g'$ and $\lnot g'$
  be the non-negated and the negated input gate of $C$ that correspond
  to the bit $b'$ in $\Rep^\iota(\A,\ov{a}^{(0)})$ representing the
  information whether $(\ov{c}_{\nu-1},\pi_j(\ov{c}_\nu))$ belongs to
  $R^\A$.

  \begin{itemize}
  \item We replace gate $g$ by the new input gate $\nicht w_\nu$, and
    gate $g'$ by the input gate $w_\nu$.
  \item Accordingly, the negated input gates $\nicht g$ and
    $\nicht g'$ are replaced by the input gates $w_\nu$ and
    $\nicht w_\nu$.
  \end{itemize}
\end{enumerate}

\noindent
Afterwards, we consider all non-negated input gates $g$ of $C$ that
have not yet been replaced, and we let $b$ be the bit of
$\Rep^\iota(\A,\vec{a}^{(0)})$ that is inserted at input gate $g$.  We
replace $g$ by the constant gate $\textbf{1}$ if $b=1$, and by the
constant gate \textbf{0} if $b=0$. Accordingly, the negated input gate
$\nicht g$ is replaced by \textbf{0} if $b=1$, and by \textbf{1} if
$b=0$.

It is easy to see that the resulting circuit $\tilde{C}$ does the same
on input $w\in\set{0,1}^m$ as circuit $C$ does on input
$\Rep^\iota(\A_w,\vec{a}^{(0)})$. Furthermore, $\tilde{C}$ obviously
has the same depth as $C$, and the size of $\tilde{C}$ is smaller than
or equal to the size of $C$.

\bigskip

To finish the proof of Lemma~\ref{lemma:CircuitLemma1}, it remains to show that
$\A_{w}$ satisfies the property \eqref{prop:i-cong-t}.
For each $\nu \leq m$, we let $i_{\nu} := \abs{w_{1}\ldots w_{\nu}}_{1} \mymod t$. In particular, $i_{0} = 0$.
For each map $f$ and each subset $X$ of its domain, we let $\rela{f}{X}$ denote the restriction of $f$ to $X$.
By $\id{X}$ we denote the identity map on a set $X$. For each $\nu \leq m$, we let $\id{\nu} := \id{S^{\A_{w}}_{\nu}(\vec{a}^{(0)})}$.
Here we have extended our notation for shells which was defined above with respect to $\A$ to $\A_{w}$.
We write $g \circ f$ for the composition of maps $f : X \to Y$ and $g : Y \to Z$. 
Recall the maps $\pi_{1}, \ldots, \pi_{t}$ from the definition of $\A_{w}$
and observe that, since these maps have disjoint domains and images, their union $\pi$
is a well-defined map, and that this map is an isomorphism from
$\sphere[m]{\A}{\vec{a}^{(i)}}$ to $\sphere[m]{\A}{\vec{a}^{(i + 1 \mymod t)}}$,
for each $i\in [t]$. For each $\nu \leq m$, we let $\gamma_{\nu} := \rela{\pi}{N_{\nu}^{\A}(\vec a^{(0)})}$
(which is the same as $\rela{\pi}{N_{\nu}^{\A}(\vec a^{(i)})}$, for each $i\in [t]$).

By induction on $\nu \leq m$, we prove that $N_{\nu}^{\A_{w}}(\vec a^{(0)}) = N_{\nu}^{\A}(\vec a^{(i_{\nu})})$
and that there exists a map $\rho_{\nu}$ such that
  \begin{gather}
    \rho_{\nu} : \sphere[\nu]{\A_w}{\vec{a}^{(0)}} \ \isom \ \sphere[\nu]{\A}{\vec{a}^{(i_{\nu})}}.
    \label{prop:i-cong-t-induction}
  \end{gather}
  For $\nu = 0$, we let $\rho_{0} := \id{0}$. If $1 \leq \nu \leq m$, we let
  \[ \rho_{\nu} :=
  \begin{cases}
    \rho_{\nu-1} \union \id{\nu} & \text{ if $w_{\nu} = 0$,}\\
    (\gamma_{\nu-1} \circ \rho_{\nu-1}) \union \id{\nu} & \text{ if $w_{\nu} = 1$.}
  \end{cases}
  \]

  Before we show that $\rho_{\nu}$ satisfies \eqref{prop:i-cong-t-induction}, for each $\nu \leq m$,
  we show how to obtain the isomorphism $\rho$ from $(\A_w,\vec{a}^{(0)})$ to $(\A,\vec{a}^{(i_{m})})$
  required by \eqref{prop:i-cong-t} if $\rho_{\nu}$ satisfies \eqref{prop:i-cong-t-induction}.
  We let $\rho := \rho_{m} \union \id{A\setminus N^{\A_{w}}_{m}(\vec a^{(0)})}$, which is
  a bijection since $\rho_{m}$ is a bijection from $N^{\A_{w}}_{m}(\vec a^{(0)})$ to $N^{\A}_{m}(\vec a^{(i_{m})}) = N^{\A_{w}}_{m}(\vec a^{(0)})$.
  According to \eqref{prop:i-cong-t-induction}, the restriction of $\rho$ to $N^{\A_{w}}_{m}(\vec a^{(0)})$
  is an isomorphism. The restriction of $\rho$ to $A \setminus N^{\A_{w}}_{m}(\vec a^{(0)})$
  is an isomorphism since the definition of $\A_{w}$ (according to its case (i))
  includes a tuple of elements which all have distance greater than $m$ to $\vec a^{(0)}$
  in a relation of $\A_{w}$ iff it is included in the corresponding relation of $\A$.
  To verify that $\rho$ is an isomorphism from $(\A_w,\vec{a}^{(0)})$ to $(\A,\vec{a}^{(i_{m})})$, it suffices to show that it preserves the relations between elements from $S^{\A_{w}}_{m}$ and $S^{\A_{w}}_{m+1}$.
 This follows since, restricted to $S^{\A_{w}}_{m} \union S^{\A_{w}}_{m+1}$, the map $\rho$ is the identity.

 Now we prove the construction of $\rho_{\nu}$ correct. For $\nu = 0$, we have $i_{0} = 0$ and from the construction of $\A_{w}$, we see that
 $\sphere[0]{\A_w}{\vec{a}^{(0)}} = \sphere[0]{\A}{\vec{a}^{(0)}}$.
 Hence, $\id{0} : \sphere[\nu]{\A_w}{\vec{a}^{(0)}} \ \isom \ \sphere[\nu]{\A}{\vec{a}^{(i_{\nu})}}$.
Suppose now that $\nu\geq 1$. First, we consider the case that $w_{\nu} = 0$ and hence $i_{\nu} = i_{\nu -1}$.
By induction, $N_{\nu-1}^{\A_{w}}(\vec a^{(0)}) = N_{\nu-1}^{\A}(\vec a^{(i_{\nu})})$ and $\rho_{\nu-1} : \sphere[\nu-1]{\A_w}{\vec{a}^{(0)}} \ \isom \ \sphere[\nu-1]{\A}{\vec{a}^{(i_{\nu})}}$.
If $w_{\nu} = 0$,  case (ii) of the definition of $\A_{w}$
includes a tuple $\vec c$ with elements from $S^{\A_{w}}_{\nu-1} \union S^{\A_{w}}_{\nu}$
in a relation of $\A_{w}$ iff it is included in the corresponding relation in $\A$.
Hence, we obtain
$N_{\nu}^{\A_{w}}(\vec a^{(0)}) = N_{\nu}^{\A}(\vec a^{(i_{\nu})})$.
By definition of $\rho_{\nu}$ for $w_{\nu} = 0$, we have $\rela{\rho_{\nu}}{N^{\A_{w}}_{\nu-1}(\vec a^{(0)})} = \rho_{\nu-1}$
and hence $\rela{\rho_{\nu}}{S^{\A_{w}}_{\nu-1}(\vec{a}^{(0)})} = \id{\nu-1}$.
Since also $\rela{\rho_{\nu}}{S^{\A_{w}}_{\nu}(\vec{a}^{(0)})} = \id{\nu}$, 
we obtain $\rho_{\nu} : \sphere[\nu]{\A_w}{\vec{a}^{(0)}} \ \isom \ \sphere[\nu]{\A}{\vec{a}^{(i_{\nu})}}$.

Now consider the case that $w_{\nu} = 1$ and hence $i_{\nu} \equiv i_{\nu-1} +1 \mymod t$.
Since, by induction, \eqref{prop:i-cong-t-induction} holds for $\nu-1$,
we can actually apply $\gamma_{\nu-1}$ after $\rho_{\nu-1}$, and
hence $\rho_{\nu}$ is a well-defined bijection of $N_{\nu}^{\A_w}(\vec{a}^{(0)})$ and $N_{\nu}^{\A}(\vec{a}^{(i_{\nu})})$.
Furthermore, we have $\rho_{\nu-1}(\bar a_{j}) = \bar a_{(j+i_{\nu-1} \mymod t)}$, for each $j\in [t]$.
Since $\gamma_{\nu-1}$ is an isomorphism from $\sphere[\nu-1]{\A}{\vec{a}^{(i_{\nu-1})}}$ to $\sphere[\nu-1]{\A}{\vec{a}^{(i_{\nu})}}$, we have $\gamma_{\nu-1}(\bar a_{(j+i_{\nu-1} \mymod t)}) = \bar a_{(j+i_{\nu} \mymod t)}$.
Altogether, we obtain that $\rho_{\nu}$ maps each element of the tuple $\vec{a}^{(0)}$ to the element in the same position
of the tuple $\vec{a}^{(i_{\nu})}$, as required.
In order to show that \eqref{prop:i-cong-t-induction} holds for $\nu$, it remains to show for each relation $R\in \sigma$ of arity $r$ and each
tuple $\ov c \in A^{r}$ that $\ov c\in R^{\A_{w}}$ iff $\rho_{\nu}(\ov c)\in R^{\A}$.
If $\ov c$ contains only elements from $N^{\A_{w}}_{\nu-1}(\vec a^{(0)})$,
this follows since $\gamma_{\nu-1} \circ \rho_{\nu-1}$ is an isomorphism.
It remains to consider tuples $\ov{c}$ containing elements from $S^{\A_{w}}_{\nu}(\vec a^{(0)})$.
That is, $\ov{c}\in (S^{\A_{w}}_{\nu-1} \union S^{\A_{w}}_{\nu})^{r}$. As in the definition of $\A_{w}$, we assume for notational simplicity that $\ov{c} = (\ov c_{\nu-1}, \ov c_{\nu})$ where the elements of $\ov c_{\nu-1}$ and $\ov c_{\nu}$ belong to $S^{\A_{w}}_{\nu-1}$ and $S^{\A_{w}}_{\nu}$, respectively. Here, $\ov c_{\nu-1}$ could be the empty tuple.
By case (ii) of the definition of $\A_{w}$, we see that $\ov c\in R^{\A_{w}}$ iff  $(\ov c_{\nu-1},\pi^{-1}(\ov c_{\nu}))\in R^{A}$.
Since $\pi$ is, in particular, a partial isomorphism from $\A$ to $\A$ which is defined on all elements of $(\ov c_{\nu-1},\pi^{-1}(\ov c_{\nu}))$,
 the previous statement holds iff $\pi((\ov c_{\nu-1},\pi^{-1}(\ov c_{\nu}))) = (\pi(\ov c_{\nu-1}), \ov c_{\nu})\in R^{A}$.
 Since, as an isomorphism, $\pi$ preserves distances, the elements of $\pi(\ov c_{\nu-1})$ belong to $S^{\A}_{\nu-1}(\vec a^{(i)})$,
 for each $i\in [t]$. Since $\rela{\rho_{\nu-1}}{S^{\A_{w}}_{\nu-1}(\vec a^{(0)})} = \id{\nu-1}$ and $\rela{\rho_{\nu}}{S^{\A_{w}}_{\nu}(\vec a^{(0)})} = \id{\nu}$, 
 we have $\pi(\ov c_{\nu-1}) = \gamma_{\nu-1}(\rho_{\nu-1}(\ov c_{\nu-1})) = \rho_{\nu}(\ov c_{\nu-1})$
 and $\ov c_{\nu} = \rho_{\nu}(\ov c_{\nu})$.
 Hence, $(\pi(\ov c_{\nu-1}), \ov c_{\nu}) = (\rho_{\nu}(\ov c_{\nu-1}), \rho_{\nu}(\ov c_{\nu})) = \rho_{\nu}(\ov c)$
 and we are done.
\end{proof}

For the proof of Theorem~\ref{thm:CyclicShiftLocality}, we want to convert the circuit
$\tilde C$ of the previous lemma, without increasing its size too much, to a circuit which accepts exactly the bitstrings
that contain a number of ones which is divisible by $r$, where $r \geq 2$ is some factor of the number $t$ from the previous lemma.
\begin{lemma}\label{lemma:CircuitLemma2}
Let $m,d,M,t,p\in\NNpos$ with $m>9$ and $p,t\geq 2$ such that $p$ and $t$ are
coprime. 
Let $\tilde{C}$ be a $\MODp$-circuit of depth $d$ and size $M$ which 
has the property that for all words $w,w'\in\set{0,1}^m$ with
$\abs{w}_1\equiv |w'|_1\mymod t$, it accepts $w$ iff it accepts $w'$.
Furthermore, let $\tilde{C}$
accept all $w\in\set{0,1}^m$ with $\abs{w}_1\equiv 0 \mymod t$, and
reject all $w\in\set{0,1}^m$ with $\abs{w}_1\equiv 1 \mymod t$.

There is a $\MODp$-circuit $\hat{C}$ of depth $(d{+}6)$ and size $(tM {+} 2m^t)$ which,
for some factor $r\geq 2$ of $t$,
accepts exactly those bitstrings $w\in\set{0,1}^m$ where $\abs{w}_1\equiv
0 \mymod r$.
\end{lemma}
\begin{proof}
We let $b=
b_0b_1\cdots b_{t-1}$ be the bitstring of length $t$ where, for every
$j\in[t]$ we have
$b_j=1$ iff $\tilde{C}$ accepts bitstrings $w\in\set{0,1}^m$ with
$\abs{w}_1\equiv j \mymod t$.

For a bitstring $w\in\set{0,1}^m$ with $\abs{w}_0\geq t{-}1$, we let
$\pat(w) = a_0 a_1\cdots a_{t-1} \in \set{0,1}^t$ with 
$a_j=1$ iff $\tilde{C}$ accepts the bitstring obtained from $w$ by
replacing the first $j$ zeros with ones.
Note that if $\abs{w}_1\equiv i\mymod t$, then 
$\pat(w) = b_i b_{i+1} \cdots b_{t-1} b_0 \cdots b_{i-1}$.

\begin{claim} 
 There is a factor $r\geq 2$ of $t$ such that for all
 $w\in\set{0,1}^m$ with $\abs{w}_0\geq t{-}1$ we have: \ 
 $\pat(w)=b \iff \abs{w}_1\equiv 0\mymod r$.
\end{claim}
\begin{proof}
  Obviously, $\pat(w)=b$ is true for all $w\in\set{0,1}^{m}$ with $\abs{w}_{0} \geq t-1$ and $\abs{w}_{1} \equiv 0\mymod t$.
  In case that for all $w\in \set{0,1}^{m}$ with $\abs{w}_{0} \geq t-1$ we have
  \[ \pat(w)=b \ \iff \ \abs{w}_1\equiv 0\mymod t,\] we are done by choosing
$r:=t$.

In case that there is a $w\in \set{0,1}^{m}$ with $\abs{w}_{0} \geq t-1$,  $\pat(w)=b$, and $\abs{w}_1\equiv i \mymod
t$ for an $i\in\set{1,\ldots,t{-}1}$, we know that 
\ $
   b_0 b_1 \cdots b_{t-1} 
    = 
   b_i b_{i+1} \cdots b_{t-1} b_0 \cdots b_{i-1}.
$ \ 
Thus, for $x:= b_0 \cdots b_{i-1}$ and $y:= b_i\cdots b_{t-1}$ we have
$b=xy=yx$, and $x,y\in\set{0,1}^+$.

A basic result in word combinatorics (see Proposition~1.3.2 in
\cite{Lothaire}) states that
two words $x,y\in \set{0,1}^+$ commute (i.e., $xy=yx$) iff they are
powers of the same word (i.e., there is a $z\in\set{0,1}^+$ and
$\nu,\mu\in\NNpos$ such that $x=z^\nu$ and $y=z^\mu$).  
We choose $z\in\set{0,1}^+$ of minimal length such that $b=z^s$ for
some $s\in \NN$. Clearly, $|z|\geq 2$, since by assumption we have
$b_0b_1=10$. 

Since $z$ is of minimal length, it is straightforward to see that
for every  $w\in\set{0,1}^m$ with $\abs{w}_0\geq t{-}1$ we have: \
$\pat(w)=z^s \iff \abs{w}_1\equiv 0\mymod |z|$.
\end{proof}

\noindent
We choose $r$ according to the claim. Obviously, the following is true
for every $w\in\set{0,1}^m$: 
\[
 \abs{w}_1\equiv 0\mymod r \ \iff \left\{
\begin{array}{ll}
  (1) & \text{$\abs{w}_0\geq t{-}1$ \ and \ $\pat(w)=b$, \ or}\\
  (2) & \text{there is a $j\in[t{-}1]$ with $m{-}j\equiv 0\mymod r$ such that $\abs{w}_0=j$.}
\end{array} \right.
\]
To complete the proof of Lemma~\ref{lemma:CircuitLemma2},
it therefore suffices to construct circuits $C_{(1)}$ and $C_{(2)}$
testing for (1) and (2), respectively, 
and to let $\hat{C}$ be the disjunction of $C_{(1)}$ and $C_{(2)}$. For constructing these circuits, note the following:
\begin{itemize}
\item
  For each $j\leq m$, a circuit $C_j$ of depth 2 and size $\leq
  m^j$ which accepts an input $w\in\set{0,1}^m$ iff $\abs{w}_0=j$ can be
  realised via $\Oder_{I\subseteq \set{1,\ldots,m} \atop |I|=j} \big(
   \Und_{i\in I} \nicht w_i \ \und \
   \Und_{i\in\set{1,\ldots,m}\setminus I} w_i
  \big)$.

\item
  The circuit $C_{(2)}$ can be realised via $\Oder_{j\in [t-1]
    \atop j\equiv 0\mymod r} C_j$. In particular, $C_{(2)}$ has depth
  2 and size $< m^t$.

\item
  A circuit $C_{\geq t{-}1}$ of depth 2 and size $\leq m^{t-1}$ which accepts an input $w\in\set{0,1}^m$
  iff $\abs{w}_0\geq t{-}1$ can be realised via $\Oder_{I\subseteq
    \set{1,\ldots,m} \atop |I|=t-1} \Und_{i\in I} \nicht w_i$.

\item
  It is not difficult to construct, for each $j\in[t]$ a
  circuit $C^{(j)}$ with
   $2m$ output gates $w^{(j)}_1,\nicht
   w^{(j)}_1,\ldots,\allowbreak w^{(j)}_m,\nicht w^{(j)}_m$, such that on input of
  a $w\in\set{0,1}^m$ with $\abs{w}_0\geq t{-}1$, this circuit produces at
  its output gates $w^{(j)}_1,\ldots,w^{(j)}_m$ the string
  $w^{(j)}=w^{(j)}_1\cdots w^{(j)}_m$ obtained from $w$ 
  by replacing the first $j$ zeros with ones. 
  For $j\geq 1$ note that $w^{(j)}_i$ can be expressed as
  \[
     w_i
     \ \oder \ 
     \Big(
       \nicht w_i \ \ \und \Oder_{I\subseteq \set{1,\ldots,i-1} \atop |I|<j}
        \big(
          \Und_{i'\in I}\nicht w_{i'} 
          \ \ \und 
          \Und_{i'\in \set{1,\ldots,i-1}\setminus I} \!\!\! w_{i'}
        \big)
     \Big).
  \]
  Thus, the circuit $C^{(j)}$ has depth 4 and size $2m(3+m^{j-1})\leq 8m^{j}$.

\item
  For each $j\in [s]$ let $\tilde{C}^{(j)}$ be the concatenation
  of $C^{(j)}$ and $\tilde{C}$. Note that on input of a bitstring
  $w\in\set{0,1}^m$ with $\abs{w}_0\geq t{-}1$, the circuit
  $\tilde{C}^{(j)}$ computes, at its output gate, the   
  letter $a_j$ of $\pat(w)=a_0\cdots a_{t-1}$.
  Furthermore, $\tilde{C}^{(j)}$ has depth $(4{+}d)$ and size $\leq 8m^j
  + M$.

\item
  The circuit $C_{(1)}$ can now be realised via
  \[
     C_{\geq t-1} \ \ \und \ \ 
     \Und_{j\in [t]\atop b_j=1} \tilde{C}^{(j)} 
     \ \ \und \ \ 
     \Und_{j\in [t] \atop b_j=0} \nicht\tilde{C}^{(j)}.
  \] 
  This circuit has depth $(d{+}5)$ and size $\leq
  1+m^{t-1}+\sum_{j=0}^{t-1}(8m^j+M) \leq tM+m^t$ (for $m>9$).

\item
  Finally, the circuit $\hat{C}$ is realised via $C_{(1)}\oder
  C_{(2)}$.
  Thus, $\hat{C}$ has depth $(d{+}6)$ and size $\leq tM+2m^t$.
\end{itemize}

\end{proof}

\noindent
We are now ready for the proof of Theorem~\ref{thm:CyclicShiftLocality}.

\begin{proofc}{Theorem~\ref{thm:CyclicShiftLocality}} \ 
Let $q$ be a $kt$-ary query defined on $\Class$ by an $\arbinvFOMODpC(\sigma)$-formula
$\varphi(\ov{x}_0,\ldots,\ov{x}_{t-1})$, where $\ov{x}_i$ is a
$k$-tuple of variables, for each $i\in [t]$.
By Theorem~\ref{thm:formula2circuit}, there exist
numbers $d,s\in\NN$ such that for every $n\in\NNpos$ there is a
$\MODp$-circuit $C_{n}$ with $\lambda^\sigma_{kt}(n)$ input bits, depth $d$, and size
$n^s$ such that the following
is true for all $\sigma$-structures $\A\in\Class$ with $|A|=n$, all $k$-tuples
$\ov{a}_0,\ldots,\ov{a}_{t-1}\in A^k$, 
and all embeddings $\iota$ of $\A$ into $[n]$: 
\begin{equation*}
 C_{n} \text{ accepts } \Rep^\iota(\A,\ov{a}_0,\ldots,\ov{a}_{t-1}) 
 \ \iff \ 
 \A^\iota \models\varphi[\ov{a}_0,\ldots,\ov{a}_{t-1}]
 \ \iff \ 
 \A\models\varphi[\ov{a}_0,\ldots,\ov{a}_{t-1}].
\end{equation*}

\noindent
For contradiction, assume that for every $c\in\NN$ the query $q$ 
is \emph{not} shift $(\log n)^c$-local w.r.t.\ $t$ on $\Class$. 
Thus, in particular for $c:=2\ell (d{+}6)$ (with $\ell$ chosen as in Theorem~\ref{thm:smolensky}), 
we obtain that for all $n_0\in\NN$ there is an $n\geq n_0$, and
a $\sigma$-structure $\A\in\Class$ with $|A|=n$, and $k$-tuples
$\ov{a}_0,\ldots,\ov{a}_{t-1}\in A^k$ such that  
for $m:=(\log n)^c = (\log n)^{2\ell (d+6)}$ we have:
\begin{itemize}
 \item
  $(\NS_m^\A(\ov{a}_i),\ov{a}_i)\isom (\NS_m^\A(\ov{a}_j),\ov{a}_j)$
  \ and \ 
  $N_m^\A(\ov{a}_i) \cap N_m^\A(\ov{a}_j)=\emptyset$ \ 
  for all $i,j\in [t]$ with $i\neq j$, \ and
 \item
  $\A\models\varphi[\ov{a}_0,\ov{a}_1,\ldots,\ov{a}_{t-1}]$ \ and \ 
  $\A\not\models\varphi[\ov{a}_1,\ldots,\ov{a}_{t-1},\ov{a}_0]$.
 \end{itemize}

\noindent
We fix $n\in\NN$ sufficiently large such that, for $\hat{d}:=(d{+}6)$ and
$\varepsilon$ and $m_{\hat{d}}$ chosen as in
Theorem~\ref{thm:smolensky}, we have for $m=(\log n)^c$ that $m>9$, $m\geq m_{\hat{d}}$, and 
$n^{\varepsilon \log n}> tn^s + 2(\log n)^{ct}$.

Clearly, $C_n$ is a $\MODp$-circuit with $\lambda^\sigma_{kt}(n)$
input bits which, for every $i\in [t]$ and all embeddings $\iota_1$
and $\iota_2$ of $\A$ accepts $\Rep^{\iota_1}(\A,\vec{a}^{(i)})$ iff it
accepts $\Rep^{\iota_2}(\A,\vec{a}^{(i)})$. 
Furthermore, $C_n$ accepts $\Rep^{\iota}(\A,\vec{a}^{(0)})$ and rejects
$\Rep^{\iota}(\A,\vec{a}^{(1)})$, for every embedding $\iota$ of $\A$.
Thus, from Lemma~\ref{lemma:CircuitLemma1} we obtain a $\MODp$-circuit
$\tilde{C}$ on $m$ input bits, with depth $d$ and
size $n^s$, such that $\tilde{C}$ has the property that for all words
$w,w'\in\set{0,1}^m$ with $\abs{w}_1\equiv |w'|_1\mymod t$, it accepts $w$
iff it accepts $w'$. Furthermore, $\tilde{C}$ accepts all
$w\in\set{0,1}^m$ with $\abs{w}\equiv 0\mymod t$ and rejects all
$w\in\set{0,1}^m$ with $\abs{w}\equiv 1\mymod t$.
From Lemma~\ref{lemma:CircuitLemma2}, we therefore obtain a $\MODp$-circuit
$\hat{C}$ of depth $\hat{d}:=(d{+}6)$ and size $(tn^s+2m^t)=
(tn^s+2(\log n)^{ct})$ which,
for some factor $r\geq 2$ of $t$, accepts exactly those bitstrings
$w\in\set{0,1}^m$ where $\abs{w}_1\equiv 0\mymod r$.

Since $p$ and $t$ are coprime by assumption, and $r\geq 2$ is a factor
of $t$, we know that $r$ has a prime factor different from $p$'s
prime factor.
Thus,
from Theorem~\ref{thm:smolensky} (for $\varepsilon,\ell,m_{\hat{d}}$ chosen
as in Theorem~\ref{thm:smolensky}, and for $m\geq m_{\hat{d}}$) we
know that the size $(tn^s+2(\log n)^{ct})$ of
$\hat{C}$ must be bigger than $2^{\varepsilon
  \sqrt[\ell \hat{d}]{m}}$.
However, we chose $m= (\log n)^c=(\log n)^{2\ell \hat{d}}$, and hence
$2^{\varepsilon \sqrt[\ell \hat{d}]{m}} =2^{\varepsilon\cdot (\log n)^2} = n^{\varepsilon\cdot
  \log n} > tn^s+ 2(\log n)^{ct}$ for all sufficiently large $n$ --- a contradiction!
Thus, the proof of
Theorem~\ref{thm:CyclicShiftLocality} is complete.
\end{proofc}

\subsection{Applications}\label{subsection:Application}

In the same way as Gaifman locality (cf., e.g., \cite{Libkin2004}), also
shift locality can be used for showing that certain queries are not
expressible in particular logics.  The first example query we consider
here is the reachability query $\textit{reach}$ which associates, with
every finite directed graph $\A=(A,E^\A)$, the relation \ $
\textit{reach}(\A) := \setc{(a,b)}{\text{$\A$ contains a directed path
    from node $a$ to node $b$}}.  $

\begin{proposition}\label{prop:reach}
  Let $\sigma=\set{E}$ consist of a binary relation symbol $E$.  Let
  $p,t\in \NN$ with $p,t\geq 2$ be such that every $t$-ary query $q$
  definable in $\arbinvFOMODp(\sigma)$ is shift $f_q(n)$-local w.r.t.\
  $t$, for a function $f_q:\NN\to\NN$ where $f_q(n)\leq
  (\frac{n}{2t}{-}\frac{1}{2})$ \ for all sufficiently large $n$.
  Then, the reachability query is not definable in
  $\arbinvFOMODp(\sigma)$.
\end{proposition}
\begin{proof}
  Assume, for contradiction, that $\textit{reach}$ is definable by an
  $\arbinvFOMODp(\sigma)$-formula $\varrho(x,y)$.  Then, \ $
  \psi(x_0,\ldots,x_{t-1}) := \ \varrho(x_0,x_1) \,\und\,
  \varrho(x_1,x_2)\, \und \cdots \und\, \varrho(x_{t-2},x_{t-1}) $ \
  is an $\arbinv$\allowbreak$\FOMODp(\sigma)$-formula expressing in a
  finite directed graph $\A$, that there is a directed path from node $x_i$ to
  node $x_{i+1}$, for every $i\in [t{-}1]$.  Let $q$ be the $t$-ary
  query defined by $\psi(x_0,\ldots,x_{t-1})$.  By assumption, this
  query is shift $f_q(n)$-local w.r.t.\ $t$, for a function $f_q$ with
  $f_{q}(n)\leq \frac{n}{2t}{-}\frac{1}{2}$ for all sufficiently large
  $n$.

  Now, consider for each $\ell\in\NNpos$ the graph $\A_\ell$
  consisting of a single directed path $v_1\to v_2 \to\cdots \to
  v_{t(2\ell+1)}$ on $t{\cdot}(2\ell{+}1)$ nodes.  For each $i\in [t]$
  let $a_i:=v_{i(2\ell+1)+(\ell+1)}$. Then, the $\ell$-neighbourhoods
  of the $a_i$, for $i\in [t]$, are pairwise disjoint and isomorphic.
  The cardinality of $\A_\ell$ is $n:=t{\cdot}(2\ell{+}1)$, and thus
  $\ell=\frac{n}{2t}{-}\frac{1}{2}\geq f_{q}(n)$.  Since $q$ is shift
  $f_q(n)$-local w.r.t.\ $t$, we obtain that
  $\A_\ell\models\psi[a_0,a_1,\ldots,a_{t-1}] \iff
  \A_\ell\models\psi[a_1,\ldots,a_{t-1},a_0]$.  But in $\A_\ell$ there
  is a directed path from $a_i$ to $a_{i+1}$ for every $i\in[t{-}1]$,
  but there is no directed path from $a_{t-1}$ to $a_0$. According to
  the choice of $\psi$, we have that
  $\A_\ell\models\psi[a_0,a_1,\ldots,a_{t-1}]$ but
  $\A_\ell\not\models\psi[a_1,\ldots,a_{t-1},a_0]$ --- a
  contradiction!
\end{proof}

As an immediate consequence of Proposition~\ref{prop:reach} and
Theorem~\ref{thm:CyclicShiftLocality} we obtain (the known fact) that
the reachability query is not definable in $\arbinvFOMODp(\set{E})$,
for any prime power $p$. Proposition~\ref{prop:reach} also shows why it can be expected to be difficult to generalise Theorem~\ref{thm:weakGaifmanLocality} and
Theorem~\ref{thm:CyclicShiftLocality} from prime powers $p$ to
arbitrary numbers $p\geq 2$:
\begin{remark}\label{remark:circuitcomplexity}
  Assume, we could generalise Theorem~\ref{thm:CyclicShiftLocality}
  from prime powers $p$ to arbitrary numbers $p\geq 2$.  By
  Proposition~\ref{prop:reach}, we would then obtain that the
  reachability query is not definable in $\arbinvFOMODp(\set{E})$, for
  any $p\in\NN$ with $p\geq 2$.  The ``opposite direction'' of
  Theorem~\ref{thm:formula2circuit}, obtained in \cite{BIS90}, would
  then tell us that the reachability query is not computable in
  $\ACCO$. Here, $\ACCO=\bigcup_{p\geq 2}\ACO[p]$, where $\ACO[p]$ is
  the class of all problems computable by a family of constant depth,
  polynomial size $\MODp$-circuits.  Since the reachability query can
  be computed in nondeterministic logarithmic space, we would thus
  obtain that $\NLOGSPACE\not\subseteq \ACCO$.  This would
  constitute a major breakthrough in computational complexity: The
  current state-of-the-art (see \cite{Williams-survey} for a recent
  survey) states that $\NEXP \not \subseteq \ACCO$, but does not know
  a problem in $\PTIME$ that provably does not belong to
  $\ACCO$.

  Similarly, a generalisation of Theorem~\ref{thm:weakGaifmanLocality}
  to all odd numbers $p$ would imply that the reachability query is
  not definable in $\ACO[p]$, for any odd number $p$. Also this is
  currently not known.
\end{remark}
The fact that $\FO$-definable queries are Gaifman local
simplifies and unifies many non-expressbility results
for $\FO$. In fact, to someone acquainted with Gaifman locality, the proof
of a non-expressibility result by a locality argument can often be communicated in a simple
visual way, while the complicated combinatorial arguments in game based non-expressibility results are hidden
within the proof of Gaifman locality for $\FO$.
We give some further examples of well-known queries which have already been used in a similar context in the literature on locality, and we show that shift locality
can also be used to obtain non-expressibility of these queries in $\arbinvFOMODp(\set{E})$,
for prime powers $p$.
For none of these examples, the fact that they are not expressible in
$\arbinvFOMODp(\set{E})$, for any prime power $p$,
is new. 
That is, all examples could be proved directly using Smolensky's theorem.
But we believe that our examples show that the notion of shift locality serves
a similar purpose as Gaifman locality.
That is, the locality argument simplifies and unifies the proofs, i.e. all details of the necessary reductions are done in Theorem~\ref{thm:CyclicShiftLocality}.

Using similar constructions as in the proof of Proposition~\ref{prop:reach}, we show that none of the following
queries is definable in $\arbinvFOMODp(\set{E})$, for any prime power
$p$:\begin{itemize}\item $\textit{cycle}(\A) := \setc{a\in A}{\text{$a$ is a node that
      lies on a cycle of the graph $\A=(A,E^\A)$}}$,
\item $\textit{triangle-reach}(\A) := \setc{a\in A}{\text{$a$ is
      reachable from a triangle in the graph $\A=(A,E^\A)$}}$,
\item $\textit{same-distance}(\A) := \setc{(a,b,c)\in
    A^3}{\dist^\A(a,c)=\dist^\A(b,c)}$.
\end{itemize}In the literature, variants of these queries have previously served as examples of queries which are not Gaifman local.
In the proofs below, we will use the following notion.  Given a
directed graph $\A = (A, E^{\A})$ and an edge $e=(u,v)\in E^{\A}$, the
\emph{$\ell$-fold subdivision} of $(u,v)$ replaces the edge $(u,v)$ by
a path $u \to u_{e,1} \to \cdots \to u_{e,\ell} \to v$, where
$u_{e,1}, \ldots, u_{e,\ell} \notin A$ are new nodes introduced for
subdividing $e$.  This notion extends naturally to edge sets instead
of single edges: The $\ell$-fold subdivision of a set $E$ of directed
edges is obtained by replacing each edge $(u,v)\in E$ by a path $u\to
u_{e,1} \to \cdots \to u_{e,\ell} \to v$, where
$u_{e,1},\ldots,u_{e,\ell}$ are new nodes introduced for subdividing
the edge $e$.

\medskip

\begin{proposition}\label{prop:cycle} \ \\
  For any prime power $p$, the $\textit{cycle}$ query is not definable
  in $\arbinvFOMODp(\set{E})$.
\end{proposition}
\begin{proof}
  Let $t\in \NN$ with $t\geq 2$ such that $t$ and $p$ are coprime
  (e.g., $t=2$ if $p$ is odd, and $t=3$ if $p$ is even).  Assume, for
  contradiction, that the \textit{cycle} query is definable by an
  $\arbinvFOMODp(\set{E})$-formula $\varrho(x)$.  Let
  $x_0,\ldots,x_{t-1}$ be first-order variables that do not occur in
  $\varrho(x)$ and let
  \[
  \psi(x_0,\ldots,x_{t-1}) \ := \ \ \varrho(x_0) \ \und \Und_{1\leq
    i<t} x_i{=}x_i.
  \]
  Let $q$ be the $t$-ary query defined by $\psi(x_0,\ldots,x_{t-1})$.
  By Theorem~\ref{thm:CyclicShiftLocality}, the query $q$ is shift
  $(\log n)^{c}$-local, for some $c\in \NN$.
      \parpic[r]{\includegraphics{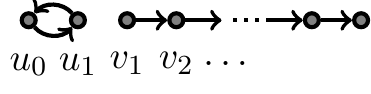}}
        Let $\A_{0}$ be a directed graph which is a disjoint union of a
  directed cycle $u_{0} \to u_{1} \to u_{0}$ and a directed path
  $v_{0} \to v_{1} \to \cdots \to v_{t}$.
       \parpic[l]{\includegraphics{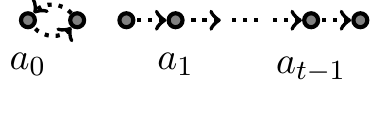}}
      
  \noindent For each $\ell \in \NNpos$, let $\A_{\ell}$ be obtained by the
  $\ell$-fold subdivision of the edge set of $\A_{0}$. In particular,
  $\A_\ell$ has cardinality $n_\ell=t{+}3 + (t{+}2)\ell$.  Let
  $a_{0}:=u_{0}$, and for each $i\in \set{1, \ldots, t{-}1}$, let
  $a_{i}:=v_{i}$. For sufficiently large values of $\ell$, the $(\log
  n_\ell)^c$-neighbourhoods of the nodes $a_0, \ldots, a_{t-1}$ in
  $\A_{\ell}$ are pairwise disjoint and isomorphic, each of them being
  isomorphic to a directed path of length $2(\log n_\ell)^c{+}1$ with
  node $a_i$ in the middle.
  
  Since $q$ is shift $(\log n)^c$-local w.r.t.\ $t$, we obtain for all
  sufficiently large $\ell$ that
  \[
  \A_\ell\models\psi[a_0,a_1,\ldots,a_{t-1}] \iff
  \A_\ell\models\psi[a_1,\ldots,a_{t-1},a_0].
  \]
  However, $a_{0}\in \textit{cycle}(\A_{\ell})$, because $a_{0}$ lies
  on a cycle, and $a_{1}\notin \textit{cycle}(\A_{\ell})$, because
  $a_{1}$ does not lie on a cycle. That is, $\A_\ell\models\varrho[a_0]$
  and $\A_\ell\not\models\varrho[a_1]$.  Thus, according to the choice
  of the formula $\psi$, we have that
  $\A_\ell\models\psi[a_0,a_1,\ldots,a_{t-1}]$, but
  $\A_\ell\not\models\psi[a_1,\ldots,a_{t-1},a_0]$, which is a
  contradiction.
\end{proof}
\begin{proposition}\label{prop:triangle-reach} \ \\
  For any prime power $p$, the $\textit{triangle-reach}$ query is not
  definable in $\arbinvFOMODp(\set{E})$.
\end{proposition}
\begin{proof}
  Let $t\in \NN$ with $t\geq 2$ such that $t$ and $p$ are coprime.
  Assume, for contradiction, that the $\textit{triangle-reach}$ query
  is definable by an $\arbinvFOMODp(\set{E})$-formula $\varrho(x)$.
  Let $x_0,\ldots,x_{t-1}$ be first-order variables that do not occur
  in $\varrho(x)$ and let
  \[
  \psi(x_0,\ldots,x_{t-1}) \ := \ \ \varrho(x_{t-1}) \ \und \
  \Und_{0\leq i < t-1} x_i{=}x_i.
  \]
  Let $q$ be the $t$-ary query defined by $\psi(x_0,\ldots,x_{t-1})$.
  By Theorem~\ref{thm:CyclicShiftLocality}, the query $q$ is shift
  $(\log n)^{c}$-local, for some $c\in \NN$.
    \parpic[r]{\centering\includegraphics{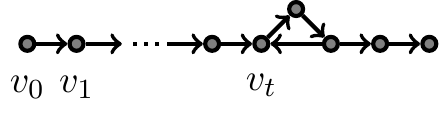}}
  Let $\A_{0}$ be the graph consisting of the directed path $v_{0} \to
  \cdots \to v_{t}$, the directed triangle $v_{t} \to v_{t+1} \to
  v_{t+2} \to v_{t}$, and the directed path $v_{t+2} \to v_{t+3} \to
  v_{t+4}$.
      
    \parpic[l]{\centering\includegraphics{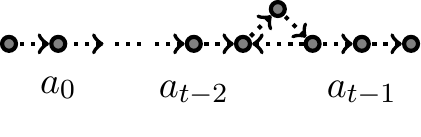}}
  For each $\ell \in \NNpos$, let $\A_{\ell}$ be obtained by the
  $\ell$-fold subdivision of all but the triangle's edges of $\A_{0}$.
  In particular, $\A_\ell$ has cardinality $n_\ell= t{+}5 +
  (t{+}2)\ell$.
  For each $i\in [t{-}1]$, let $a_{i}:=v_{i+1}$, and let
  $a_{t-1}:=v_{t+3}$.  For sufficiently large values of $\ell$, the
  $(\log n_\ell)^c$-neighbourhoods of the nodes $a_0, \ldots, a_{t-1}$
  in $\A_{\ell}$ are pairwise disjoint and isomorphic, each of them
  being isomorphic to a directed path of length $2(\log n_\ell)^c{+}1$
  with node $a_i$ in the middle.

  Since $q$ is shift $(\log n)^c$-local w.r.t.\ $t$, we obtain for all
  sufficiently large $\ell$ that
  \[
  \A_\ell\models\psi[a_0,a_1,\ldots,a_{t-1}] \iff
  \A_\ell\models\psi[a_1,\ldots,a_{t-1},a_0].
  \]
  However, $a_{t-1}\in \textit{triangle-reach}(\A_{\ell})$, because
  $a_{t-1}$ is reachable from the triangle, and $a_{0}\notin
  \textit{triangle-reach}(\A_{\ell})$, because there is no directed
  path from any node of the triangle to $a_{0}$.  That is,
  $\A_\ell\models\varrho[a_{t-1}]$ and
  $\A_\ell\not\models\varrho[a_0]$.  Thus, due to the choice of the
  formula $\psi$, we have that
  $\A_\ell\models\psi[a_0,a_1,\ldots,a_{t-1}]$, but
  $\A_\ell\not\models\psi[a_1,\ldots,a_{t-1},a_0]$, which is a
  contradiction.
\end{proof}
\begin{proposition}\label{prop:same-distance} \ \\
  For any prime power $p$, the $\textit{same-distance}$ query is not
  definable in $\arbinvFOMODp(\set{E})$.
\end{proposition}
\begin{proof}
  Let $t\in \NN$ with $t \geq 3$ such that $t$ and $p$ are coprime.
  Assume, for contradiction, that the $\textit{same-distance}$ query
  is definable by an $\arbinvFOMODp(\set{E})$-formula
  $\varrho(x,y,z)$.
  
  \noindent Let $x_0,\ldots,x_{t-1}$ be first-order variables that do
  not occur in $\varrho(x,y,z)$ and let
  \[
  \psi(x_0,\ldots,x_{t-1}) \ := \ \ \varrho(x_{t-3},x_{t-2},x_{t-1}) \
  \; \und \Und_{0\leq i<t-3} x_i{=}x_i.
  \]
  Let $q$ be the $t$-ary query defined by $\psi(x_0,\ldots,x_{t-1})$.
  By Theorem~\ref{thm:CyclicShiftLocality}, the query $q$ is shift
  $(\log n)^{c}$-local, for some $c\in \NN$.
  
\parpic[r]{\centering\includegraphics{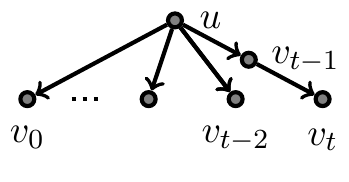}}
  Let $\A_{0}$ be the graph consisting of the edges $u \to v_{i}$ for
  each $i\in \set{0, \ldots, t{-}2}$, and a directed path $u \to
  v_{t-1} \to v_{t}$.

  For each $\ell \in \NNpos$, let $\A_{\ell}$ be the graph obtained
    by the $\ell$-fold subdivision of the edges $u \to v_{i}$ in
    $\A_{0}$, for each $i\in [t]$. 
    \parpic[l]{\centering\includegraphics{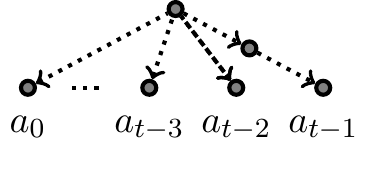}}\noindent
    In particular, $\A_\ell$ has
    cardinality $n_\ell=t{+}2+t\ell$.
    For each $i\in [t{-}1]$ let $a_{i}:=v_{i}$, and let $a_{t-1} :=
    v_{t}$.
    For sufficiently large values of $\ell$, the $(\log
    n_\ell)^c$-neighbourhoods of nodes the $a_0, \ldots, a_{t-1}$ in
    $\A_{\ell}$ are pairwise disjoint and isomorphic, each of them
    being isomorphic to a directed path of length $(\log
    n_\ell)^c{+}1$ with $a_i$ at the end of the path.

    Since $q$ is shift $(\log n)^c$-local w.r.t.\ $t$, we obtain for
    all sufficiently large $\ell$ that
    \[
    \A_\ell\models\psi[a_0,a_1,\ldots,a_{t-1}] \iff
    \A_\ell\models\psi[a_1,\ldots,a_{t-1},a_0].
    \]
    However, $(a_{t-3},a_{t-2},a_{t-1})\in
    \textit{same-distance}(\A_{\ell})$, because $a_{t-3}$ and
    $a_{t-2}$ both have distance $2(\ell{+}1){+}1$ to $a_{t-1}$ in
    $\A_{\ell}$, and $(a_{t-2},a_{t-1},a_{0})\notin
    \textit{same-distance}(\A_{\ell})$, because $a_{t-2}$ has distance
    $2(\ell{+}1)$ to $a_{0}$, but $a_{t-1}$ has distance
    $2(\ell{+}1){+}1$ to $a_{0}$.  That is,
    $\A_\ell\models\varrho[a_{t-3},a_{t-2},a_{t-1}]$ and
    $\A_\ell\not\models\varrho[a_{t-2},a_{t-1},a_0]$.  Thus, due to
    the choice of the formula $\psi$, we have that
    $\A_\ell\models\psi[a_0,a_1,\ldots,a_{t-1}]$, but
    $\A_\ell\not\models\psi[a_1,\ldots,a_{t-1},a_0]$, which is a
    contradiction.
  \end{proof}

\section{Hanf locality and locality on string structures}\label{section:strings}

For giving the precise definition of Hanf locality, we need the following notation:
As in \cite{Libkin2004}, for $\sigma$-structures $\A$ and $\B$, for $k$-tuples
$\ov{a}\in A^k$ and $\ov{b}\in B^k$, and for an $r\in \NN$, we write
$(\A,\ov{a})\Hanfequiv_r (\B,\ov{b})$ (or simply  $\A\Hanfequiv_r\B$ in
case that $k{=}0$) if there is a bijection $\beta:A\to
B$ such that 
$(\NS_r^\A(\ov{a}c),\ov{a}c)\isom (\NS_r^\B(\ov{b}\beta(c)),\ov{b}\beta(c))$
is true for every $c\in A$.

\begin{definition}[Hanf locality]
Let $\Class$ be a class of finite $\sigma$-structures, $k\in\NN$, and
$f:\NN\to\NN$.\\
A $k$-ary query $q$ is \emph{Hanf $f(n)$-local on $\Class$} if there
is an $n_0\in\NN$ such that for every $n\in\NN$ with $n\geq n_0$ and
all $\sigma$-structures $\A,\B\in\Class$ with $|A|=|B|=n$, the
following is true for all $k$-tuples $\ov{a}\in A^k$ and $\ov{b}\in
B^k$ with $(\A,\ov{a})\Hanfequiv_{f(n)} (\B,\ov{b})$: 
\quad $\ov{a}\in q(\A) \iff \ov{b}\in q(\B)$.
\\
The query $q$ is called \emph{Hanf $f(n)$-local} if it is Hanf
$f(n)$-local on the class of all finite $\sigma$-structures.
\end{definition}
Consider a $0$-ary query which maps a structure,
containing a unary relation $P$, to the relation containing the empty tuple if the number of elements contained in $P$ is even, and to the empty set, otherwise.
This is an example of a Hanf $0$-local query. The $0$-ary query
which maps a graph to the relation containing only the empty tuple if it is connected,
and to the empty set otherwise, is an example of a query which is not Hanf $f(n)$-local
for any sublinear function $f$.

Hanf locality is an even stronger locality notion than Gaifman
locality:

\begin{theorem}[Hella, Libkin, Nurmonen \cite{HellaLibkinNurmonen1999}]\label{thm:hanf2gaifman}
Let $\Class$ be a class of finite $\sigma$-structures 
and let $f:\NN\to \NN$. Let $k\in\NNpos$ and let $q$ be a $k$-ary query.
If $q$ is Hanf $f(n)$-local on $\Class$, then $q$ is 
Gaifman $(3f(n){+}1)$-local on $\Class$.
\end{theorem}

It is well-known that queries definable in $\FO$ or $\FOMODp$ (for any
$p\geq 2$) are Hanf local with a constant locality radius \cite{FSV95,HellaLibkinNurmonen1999}. 
For order-invariant or arb-invariant $\FO$ it is still open whether
they are Hanf local with respect to any sublinear locality radius.
As an immediate consequence of Proposition~\ref{prop:notGaifmanLocal}
and Theorem~\ref{thm:hanf2gaifman} one obtains for every $p\in\NN$
with $p\geq 2$ that order-invariant $\FOMODp$ is \emph{not} Hanf local with
respect to any sublinear locality radius.

An $\arbinvFOMODp(\sigma_\Sigma)$-sentence $\varphi$ defines
the string-language $L_\varphi:=\setc{w\in\Sigma^+}{\S_w\models\varphi}$.
A string-language $L\subseteq \Sigma^+$ corresponds to a 0-ary query
$q_L$ on $\ClassStrings$ defined, for each $w\in\Sigma^+$, via \
\[
   q_L(\S_w) \ := \ \left\{
   \begin{array}{cl}
     \set{()} & \text{if } \ w\in L \\
     \emptyset & \text{if } \ w\not\in L,
   \end{array}
   \right.
\]
where $()$ denotes the unique tuple of arity 0.
A string-language $L\subseteq \Sigma^+$ is called \emph{Hanf
  $f(n)$-local} iff the 0-ary query $q_L$ is Hanf $f(n)$-local on
$\ClassStrings$, i.e., there is an $n_0\in\NN$ such that for every
$n\in\NN$ with $n\geq n_0$ and all strings
$u,v\in\Sigma^+$ of length $n$ with
$\S_u\Hanfequiv_{f(n)} \S_{v}$, we have: \ $u\in L \iff v\in L$.

For the restricted case of string structures, Benedikt and Segoufin
\cite{BenediktSegoufin2009a} have shown that on $\ClassStrings$ order-invariant
$\FO$ has the same expressive power as $\FO$ and
thus is Hanf local with constant locality radius (in fact,
\cite{BenediktSegoufin2009a} obtains the same result also for
finite labelled ranked trees). 
In \cite{AMSS-SICOMP} it was shown that every query 
definable in arb-invariant $\FO$ on $\ClassStrings$ is Hanf local with
polylogarithmic locality radius, and that in the worst case the
locality radius can indeed be of polylogarithmic size.
As an immediate consequence of
Proposition~\ref{prop:notWeaklyGaifmanLocal} and
Theorem~\ref{thm:hanf2gaifman}
we obtain that for $\Sigma:=\set{0,1}$ there is a unary query $q$
that is \emph{not} Hanf $(\frac{n-8}{12})$-local on $\ClassStrings$, but
definable in $\ordinvFOMODp(\sigma_\Sigma)$ for every \emph{even}
number $p\geq 2$.
From this observation, we can also obtain an example of a language 
which is not Hanf $(\frac{n-8}{12})$-local on strings over an extended alphabet.
The existence of this language can be obtained using a general principle 
which allows to convert a $k$-ary query $q$ over strings to a $0$-ary query, i.e. a language $A_{q}$ over strings over an extended alphabet.
The language $A_{q}$ inherits the relevant definability and locality properties for our purposes from the query $q$.
This principle is stated in Lemma~\ref{lem:string-query-arity-reduction} below, since we will need it for the proof of Theorem~\ref{thm:HanfLocalityOnStrings} below.
Using this approach here, however, it becomes rather hard to describe the shape of the language 
which is not Hanf $(\frac{n-8}{12})$-local concretely.
To this end, we show directly how this language can be obtained by 
simple modification of the proof of Proposition~\ref{prop:notWeaklyGaifmanLocal}.
 Consider the languages
\begin{eqnarray*}
   L_{\textit{left}} & := &1^{+}\,2\,0^{+}1^{+}0^{+}\\
   L_{\textit{right}} & := &1^{+}0^{+}1^{+}\,2\,0^{+}
\end{eqnarray*}
over the alphabet $\Sigma:=\set{0,1,2}$.
Note that the definitions of $L_{\textit{left}}$ and $L_{\textit{right}}$ are very similar, the only difference
being the position of the unique $2$ occurring in strings from $L_{\textit{left}}$ and $L_{\textit{right}}$.
We define
\begin{eqnarray*}
  L_{\textit{even}} & := & \setc{w\in L_{\textit{right}}}{|w|_1 \mymod 2 \equiv 0 }\\
  L_{\textit{odd}} & := & \setc{w\in L_{\textit{left}}}{|w|_1 \mymod 2 \equiv 1 }
\end{eqnarray*}

\smallskip

\begin{proposition} 
  Let $L := L_{\textit{even}} \union L_{\textit{odd}}$.
 \begin{enumerate}[label=(\emph{\alph*})]
  \item
   The language $L$ is not Hanf $(\frac{n-1}{8})$-local.
  \item
   $L$ is definable by a 
   sentence $\varphi$ in
  $\ordinvFOMODp(\sigma_\Sigma)$, for every even number $p\geq
  2$.
 \end{enumerate}
\end{proposition}
\begin{proof}
\emph{(a):}  \ For every $\ell\in\NNpos$ let
\[
  \begin{array}{rcrclcl}
   u_\ell & := & 
     1^\ell \ 1^\ell 2 0^\ell \ 0^\ell 1^\ell \ 1^\ell 0^\ell\ 0^\ell
    & = &
    \stringfont{x} \; \stringfont{y}_1 \; \stringfont{y}_2 \;
    \stringfont{y}_3 \; \stringfont{z},
   \\[1ex]
   v_\ell & := & 
     1^\ell \ 1^\ell 0^\ell \ 0^\ell 1^\ell \ 1^\ell 2 0^\ell \ 0^\ell
    & = &
    \stringfont{x} \; \stringfont{y}_3 \; \stringfont{y}_2 \; \stringfont{y}_1\;\stringfont{z},
  \end{array}
\]
for \quad
$
  \stringfont{x}   := 1^\ell, \quad
  \stringfont{y}_1 := 1^\ell 2 0^{\ell},\quad
  \stringfont{y}_2 := 0^\ell 1^{\ell}, \quad
  \stringfont{y}_3 := 1^\ell 0^\ell, \quad
  \stringfont{z}   := 0^\ell.
$

 \medskip
 \noindent
It is not difficult to see that $\S_{u_\ell}\Hanfequiv_\ell
\S_{v_\ell}$: the bijection $\beta$, for which
\[
   (\NS_\ell^{\S_{u_\ell}}(c),c)
   \ \isom \
   (\NS_\ell^{\S_{v_\ell}}(\beta(c)),\beta(c)) \qquad
   \text{for every \ } c\in \set{1,\ldots,|u_\ell|},
\]
can be chosen as follows. It maps each position of
\begin{itemize}
 \item
   $\stringfont{x}$ in $u_\ell$ onto the according position 
   of $\stringfont{x}$ in $v_\ell$,
 \item 
   $\stringfont{y}_s$ (for $s\in\set{1,2,3}$) in $u_\ell$ onto the according position 
   of $\stringfont{y}_s$ in $v_\ell$,
 \item 
   $\stringfont{z}$ in $u_\ell$ onto the according position 
   of $\stringfont{z}$ in $v_\ell$.
\end{itemize}

\noindent
It is straightforward to verify that this bijection $\beta$ indeed witnesses that
$\S_{u_\ell}\Hanfequiv_\ell \S_{v_\ell}$. Furthermore, $v_\ell\in
L$ and $u_\ell\not\in L$. The length of $u_\ell$ and
$v_\ell$ is $n:=8\ell{+}1$, thus
$\ell=\frac{n-1}{8}$. Hence, the language $L$ is not Hanf
$(\frac{n-1}{8})$-local.
This completes the proof of \emph{(a)}.

\bigskip

\noindent
\emph{(b):} \ 
First, note that the language
\[
M \ := \ \
L_{\textit{left}}\ \cup \ L_{\textit{right}}
\]
can be defined by an $\FO(\sigma_\Sigma)$-sentence $\varphi_M$ which
states the following:

\begin{itemize}
 \item
   The first position of the string carries the letter 1. The last
   position of the string carries the letter 0.
 \item 
   For each position $x$ that carries the letter 1, the position
   directly to the right of $x$ carries one of the letters 0, 1, 2.
   Furthermore, there is exactly one position $x$ that carries the
   letter 1, such that the position directly to the right of $x$
   carries the letter 0. And there is exactly one position $x$ that
   carries the letter 1, such that the position directly to the right
   of $x$ carries the letter 2.
 \item 
   For each position $y$ that carries the letter 0, the position
   directly to the right of $y$ carries one of the letters 0, 1.
   Furthermore, there is exactly one position $y$ that carries the
   letter 0, such that the position directly to the right of $y$
   carries the letter 1.
 \item
   There is exactly one position $z$ that carries the letter 2. The
   position directly to the right of $z$ carries the letter 0.
\end{itemize}

From Example~\ref{example:niemistoe-evencycles} we obtain an
$\ordinvFOMODnum{2}(\set{E})$-sentence
$\varphi_{\textit{even cycles}}$ that is satisfied by a finite
$\set{E}$-structure $\A$ iff $\A$ is a disjoint union of directed
cycles where the number of cycles of even length is even.
We choose
\[
   \varphi \ := \ \ \varphi_M \,\und\,\varphi',
\]
where $\varphi'$ is the formula obtained from $\varphi_{\textit{even
    cycles}}$ by 
relativisation of all quantifiers to the positions that carry the letters
1 or 2, and by
replacing every atom of the form $E(\mu,\nu)$ (for first-order
variables $\mu$ and $\nu$) by a formula stating that
\begin{itemize}  
 \item $E(\mu,\nu)$ is true, or
 \item position $\mu$ carries the letter 2, and $\nu$ is the leftmost
   position of the string, or
 \item the positions $\mu$ and $\nu$ both carry the letter 1, and the
   positions directly to the right of $\mu$ and directly to the left
   of $\nu$ both carry the letter 0. 
\end{itemize}\smallskip

\noindent Clearly, the obtained formula is order-invariant on the class of all
finite $\sigma_\Sigma$-structures, since the formula $\varphi_{\textit{even
    cycles}}$ is order-invariant on the class of all finite $\set{E}$-structures.

It is straightforward to see that, when evaluated in a string $w\in L_{\textit{left}}$ of the
form $1^i20^j1^{i'}0^{j'}$, the formula
$\varphi'$ simulates the evaluation of the formula
$\varphi_{\textit{even cycles}}$ in a graph that consists of the
disjoint union of two cycles of lengths $i{+}1$ and $i'$; and hence
$\varphi'$ is satisfied iff $i{+}1$ and $i'$ are either both even or both
odd --- and this is equivalent to the statement that $|w|_1 = i{+}i'$ is odd.
Similarly, when evaluated in a string $w\in L_{\textit{right}}$ of the form $1^i0^j1^{i'}20^{j'}$,
the formula $\varphi'$ simulates the
evaluation of the formula $\varphi_{\textit{even cycles}}$ in a graph
that consists of a single cycle of length $i{+}i'{+}1$; and hence
$\varphi'$ is satisfied iff $i{+}i'{+}1$ is odd, i.e., $|w|_1 = i{+}i'$ is even.

In summary, we obtain that $\varphi$ is an
$\ordinvFOMODnum{2}(\sigma_\Sigma)$-sentence that defines the language
$L$.
Since modulo 2 counting quantifiers can be simulated by modulo $p$
counting quantifiers, for every \emph{even} number $p\geq 2$, the
proof of (b) is complete.
\end{proof}

\begin{proposition} 
  Let $L := L_{\textit{even}} \union L_{\textit{odd}}$.
 \begin{enumerate}[label=\({\alph*}]
  \item
   The language $L$ is not Hanf $(\frac{n-1}{8})$-local.
  \item
   $L$ is definable by a 
   sentence $\varphi$ in
  $\ordinvFOMODp(\sigma_\Sigma)$, for every even number $p\geq
  2$.
 \end{enumerate}
\end{proposition}
\begin{proof}\hfill

\noindent\(a:  \ For every $\ell\in\NNpos$ let
\[
  \begin{array}{rcrclcl}
   u_\ell & := & 
     1^\ell \ 1^\ell 2 0^\ell \ 0^\ell 1^\ell \ 1^\ell 0^\ell\ 0^\ell
    & = &
    \stringfont{x} \; \stringfont{y}_1 \; \stringfont{y}_2 \;
    \stringfont{y}_3 \; \stringfont{z},
   \\[1ex]
   v_\ell & := & 
     1^\ell \ 1^\ell 0^\ell \ 0^\ell 1^\ell \ 1^\ell 2 0^\ell \ 0^\ell
    & = &
    \stringfont{x} \; \stringfont{y}_3 \; \stringfont{y}_2 \; \stringfont{y}_1\;\stringfont{z},
  \end{array}
\]
for \quad
$
  \stringfont{x}   := 1^\ell, \quad
  \stringfont{y}_1 := 1^\ell 2 0^{\ell},\quad
  \stringfont{y}_2 := 0^\ell 1^{\ell}, \quad
  \stringfont{y}_3 := 1^\ell 0^\ell, \quad
  \stringfont{z}   := 0^\ell.
$

 \medskip
 \noindent
It is not difficult to see that $\S_{u_\ell}\Hanfequiv_\ell
\S_{v_\ell}$: the bijection $\beta$, for which
\[
   (\NS_\ell^{\S_{u_\ell}}(c),c)
   \ \isom \
   (\NS_\ell^{\S_{v_\ell}}(\beta(c)),\beta(c)) \qquad
   \text{for every \ } c\in \set{1,\ldots,|u_\ell|},
\]
can be chosen as follows. It maps each position of
\begin{itemize}
 \item
   $\stringfont{x}$ in $u_\ell$ onto the according position 
   of $\stringfont{x}$ in $v_\ell$,
 \item 
   $\stringfont{y}_s$ (for $s\in\set{1,2,3}$) in $u_\ell$ onto the according position 
   of $\stringfont{y}_s$ in $v_\ell$,
 \item 
   $\stringfont{z}$ in $u_\ell$ onto the according position 
   of $\stringfont{z}$ in $v_\ell$.
\end{itemize}

\noindent
It is straightforward to verify that this bijection $\beta$ indeed witnesses that
$\S_{u_\ell}\Hanfequiv_\ell \S_{v_\ell}$. Furthermore, $v_\ell\in
L$ and $u_\ell\not\in L$. The length of $u_\ell$ and
$v_\ell$ is $n:=8\ell{+}1$, thus
$\ell=\frac{n-1}{8}$. Hence, the language $L$ is not Hanf
$(\frac{n-1}{8})$-local.
This completes the proof of \(a.

\bigskip

\noindent
\(b: \ 
First, note that the language
\[
M \ := \ \
L_{\textit{left}}\ \cup \ L_{\textit{right}}
\]
can be defined by an $\FO(\sigma_\Sigma)$-sentence $\varphi_M$ which
states the following:

\begin{itemize}
 \item
   The first position of the string carries the letter 1. The last
   position of the string carries the letter 0.
 \item 
   For each position $x$ that carries the letter 1, the position
   directly to the right of $x$ carries one of the letters 0, 1, 2.
   Furthermore, there is exactly one position $x$ that carries the
   letter 1, such that the position directly to the right of $x$
   carries the letter 0. And there is exactly one position $x$ that
   carries the letter 1, such that the position directly to the right
   of $x$ carries the letter 2.
 \item 
   For each position $y$ that carries the letter 0, the position
   directly to the right of $y$ carries one of the letters 0, 1.
   Furthermore, there is exactly one position $y$ that carries the
   letter 0, such that the position directly to the right of $y$
   carries the letter 1.
 \item
   There is exactly one position $z$ that carries the letter 2. The
   position directly to the right of $z$ carries the letter 0.
\end{itemize}

\noindent From Example~\ref{example:niemistoe-evencycles} we obtain an
$\ordinvFOMODnum{2}(\set{E})$-sentence
$\varphi_{\textit{even cycles}}$ that is satisfied by a finite
$\set{E}$-structure $\A$ iff $\A$ is a disjoint union of directed
cycles where the number of cycles of even length is even.
We choose
\[
   \varphi \ := \ \ \varphi_M \,\und\,\varphi',
\]
where $\varphi'$ is the formula obtained from $\varphi_{\textit{even
    cycles}}$ by 
relativisation of all quantifiers to the positions that carry the letters
1 or 2, and by
replacing every atom of the form $E(\mu,\nu)$ (for first-order
variables $\mu$ and $\nu$) by a formula stating that
\begin{itemize}  
 \item $E(\mu,\nu)$ is true, or
 \item position $\mu$ carries the letter 2, and $\nu$ is the leftmost
   position of the string, or
 \item the positions $\mu$ and $\nu$ both carry the letter 1, and the
   positions directly to the right of $\mu$ and directly to the left
   of $\nu$ both carry the letter 0. 
\end{itemize}

\noindent Clearly, the obtained formula is order-invariant on the class of all
finite $\sigma_\Sigma$-structures, since the formula $\varphi_{\textit{even
    cycles}}$ is order-invariant on the class of all finite $\set{E}$-structures.

It is straightforward to see that, when evaluated in a string $w\in L_{\textit{left}}$ of the
form $1^i20^j1^{i'}0^{j'}$, the formula
$\varphi'$ simulates the evaluation of the formula
$\varphi_{\textit{even cycles}}$ in a graph that consists of the
disjoint union of two cycles of lengths $i{+}1$ and $i'$; and hence
$\varphi'$ is satisfied iff $i{+}1$ and $i'$ are either both even or both
odd --- and this is equivalent to the statement that $|w|_1 = i{+}i'$ is odd.
Similarly, when evaluated in a string $w\in L_{\textit{right}}$ of the form $1^i0^j1^{i'}20^{j'}$,
the formula $\varphi'$ simulates the
evaluation of the formula $\varphi_{\textit{even cycles}}$ in a graph
that consists of a single cycle of length $i{+}i'{+}1$; and hence
$\varphi'$ is satisfied iff $i{+}i'{+}1$ is odd, i.e., $|w|_1 = i{+}i'$ is even.

In summary, we obtain that $\varphi$ is an
$\ordinvFOMODnum{2}(\sigma_\Sigma)$-sentence that defines the language
$L$.
Since modulo 2 counting quantifiers can be simulated by modulo $p$
counting quantifiers, for every \emph{even} number $p\geq 2$, the
proof of (b) is complete.
\end{proof}

\begin{remark}
  Benedikt and Segoufin \cite{BenediktSegoufin2009a} conjectured
  that order-invariant $\FOMOD{}$ (i.e. formulas with arbitrary modulo counting quantifiers) has the same expressive power as $\FOMOD{}$ on trees. The previous proposition
  shows that, for even numbers $p \geq 2$, there exist order-invariant $\FOMODp$-definable
  languages which are not $\FOMOD{}$-definable, since $\FOMOD{}$ is Hanf local.
  This refutes the conjecture even for strings instead of trees.
\end{remark}

From Niemist\"o's Corollary~7.25 in
\cite{Niemistoe-PhD} it follows that for \emph{odd} numbers $p$,
order-invariant\linebreak $\FOMODp(\sigma_\Sigma)$ on $\ClassStrings$ has
exactly the same expressive power as
$\FOMOD{\text{PFC}(p)}(\sigma_\Sigma)$, where $\text{PFC}(p)$ is the
set of all numbers whose prime factors are prime factors of $p$, and
$\FOMOD{\text{PFC}(p)}$ is first-order logic with modulo $m$ counting
quantifiers for all $m\in\text{PFC}(p)$. 

The present section's main result shows that for \emph{odd prime
  powers} $p$, the Hanf locality result of \cite{AMSS-SICOMP} can be
generalised from arb-invariant $\FO$ to arb-invariant $\FOMODp$ on
$\ClassStrings$:

\begin{theorem}\label{thm:HanfLocalityOnStrings}
Let $\Sigma$ be a finite alphabet. Let $k\in\NN$, let $q$ be a $k$-ary
query, and let $p$ be an odd prime power. If $q$ is definable in
$\arbinvFOMODp(\sigma_\Sigma)$ on $\ClassStrings$, then there
is a $c\in\NN$ such that $q$ is Hanf $(\log n)^c$-local on $\ClassStrings$.
\end{theorem}

\noindent
Together with Theorem~\ref{thm:hanf2gaifman} this implies (general instead
of weak) Gaifman
locality on $\ClassStrings$:

\begin{corollary}\label{cor:GaifmanLocalityOnStrings}
Let $\Sigma$ be a finite alphabet. Let $k\in\NNpos$, let $q$ be a $k$-ary
query, and let $p$ be an odd prime power. If $q$ is definable in
$\arbinvFOMODp(\sigma_\Sigma)$ on $\ClassStrings$, then there
is a $c\in\NN$ such that $q$ is Gaifman $(\log n)^c$-local on $\ClassStrings$.
\end{corollary}

Note that this corollary does not contradict the non-locality result of 
Proposition~\ref{prop:notGaifmanLocal}, as the counter-example given in
the proof of that proposition is not a string structure.

The remainder of this section is devoted to the proof of
Theorem~\ref{thm:HanfLocalityOnStrings}.
We follow the overall approach of \cite{AMSS-SICOMP}.
The crucial step is to prove Theorem~\ref{thm:HanfLocalityOnStrings}
for queries $q$ of arity $k=0$.
The case for queries of arity $k\geq 1$ can then easily be reduced to the case for queries of arity $0$ by adding $k$ extra symbols to the alphabet; see below.

Note that a $0$-ary query $q$ defines the string-language
$L_q:=\setc{w\in\Sigma^+}{()\in q(\S_w)}$, where $()$ denotes the unique
tuple of arity 0. 
The language $L_q$ is called \emph{Hanf $f(n)$-local} iff
$q$ is Hanf $f(n)$-local on $\ClassStrings$.
For proving Theorem~\ref{thm:HanfLocalityOnStrings} for the case
$k=0$, we consider the following notion.

\begin{definition}[Disjoint swaps \cite{AMSS-SICOMP}]\label{def:disjoint-swaps}
 Let $r\in\NN$ and let $w\in\Sigma^+$ be a string over a finite
 alphabet $\Sigma$. 
 A string $w'\in\Sigma^+$ is obtained from $w$ by a \emph{disjoint
   $r$-swap operation} if 
 there exist strings
 $\stringfont{x},\stringfont{u},\stringfont{y},\stringfont{v},\stringfont{z}$
 such that  
 $w=\stringfont{xuyvz}$ and $w'=\stringfont{xvyuz}$, and for the positions
 $i,j,i',j'$ of $w$ just before
 $\stringfont{u},\stringfont{y},\stringfont{v},\stringfont{z}$ the
 following is true: \
 The neighbourhoods
 $N^{\S_w}_r(i)$, $N^{\S_w}_r(j)$,  $N^{\S_w}_r(i')$,
 $N^{\S_w}_r(j')$ are pairwise disjoint, and 
 $(\NS^{\S_w}_r(i),i)\isom (\NS^{\S_w}_r(i'),i')$ and
 $(\NS^{\S_w}_r(j),j)\isom (\NS^{\S_w}_r(j'),j')$.

 Let $f:\NN\to\NN$. A string-language $L\subseteq \Sigma^+$ is
 \emph{closed under disjoint $f(n)$-swaps} if there exists an
 $n_0\in\NN$ such that for every $n\in\NNpos$ with $n\geq n_0$, all
 strings $w\in\Sigma^+$ of length $n$, and all strings $w'$ obtained
 from $w$
 by a disjoint $f(n)$-swap operation, we have: \ $w\in L \iff w'\in L$.
\end{definition}
  \begin{figure}[t]
    \centering
    \resizebox{250pt}{!}{\includegraphics{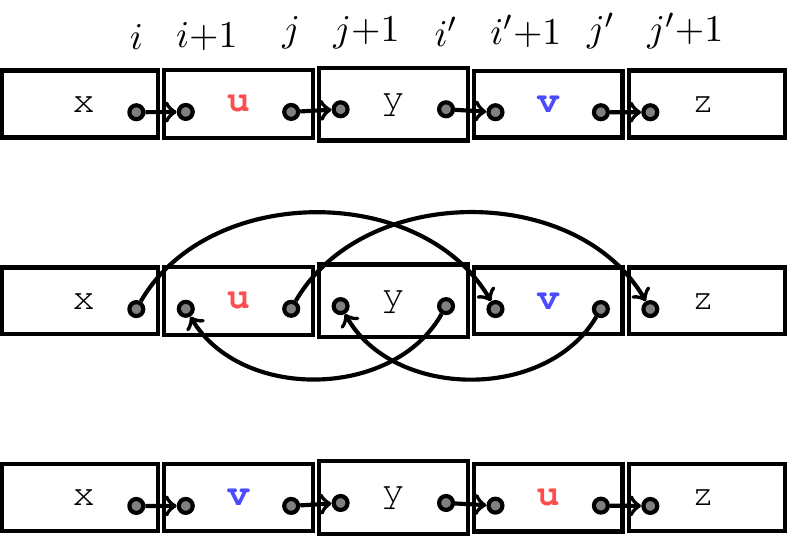}}
\caption{The picture shows how the disjoint $r$-swap operation turns the edge relation
  of the string structure of a string $w = \stringfont{xuyvz}$ (top) into the edge relation of the string structure of $w' = \stringfont{xvyuz}$ (middle and bottom).}

\end{figure}

It was shown in \cite{AMSS-SICOMP} (see Proposition~5.7, Lemma~5.2,
and the proof of Theorem~5.1 in \cite{AMSS-SICOMP}) 
that if a language $L\subseteq \Sigma^+$ is
closed under 
disjoint $(\log n)^d$-swaps, for some $d\in\NN$, then
it is Hanf $(\log n)^c$-local, for some $c>d$.
Hence, the following lemma immediately implies
Theorem~\ref{thm:HanfLocalityOnStrings} for the case $k=0$.

\begin{lemma}\label{lemma:disjoint-swaps}
  Let $\Sigma$ be a finite alphabet, let $L\subseteq \Sigma^{+}$, and
  let $p$ be an odd prime power.
  If $L$ is definable by an $\arbinvFOMODp(\sigma_\Sigma)$-sentence,
  then there exists a constant $d\in\NN$
  such that $L$ is closed under disjoint $(\log n)^d$-swaps.
\end{lemma}

\begin{proof}
 We proceed in the same way as in the proof of Proposition~5.5
 in \cite{AMSS-SICOMP}, which obtained
 the analogue of Lemma~\ref{lemma:disjoint-swaps} for
 $\arbinvFO(\sigma_\Sigma)$-sentences.
 However, we cannot just copy the proof from there, since that proof
 relies on (general) Gaifman locality (with polylogarithmic locality
 radius) of queries definable in $\arbinvFO(\sigma_\Sigma)$, 
 while in the present case we have available (from
 Theorem~\ref{thm:weakGaifmanLocality}) only \emph{weak} Gaifman 
 locality (with polylogarithmic locality radius) of queries definable
 in $\arbinvFOMODp(\sigma_\Sigma)$.

 Let $\varphi$ be an $\arbinvFOMODp(\sigma_\Sigma)$-sentence
 which defines a string language
 $L_\varphi=\setc{w\in\Sigma^+}{\S_w\models\varphi}$.
 For contradiction, assume that there is no $d\in\NN$ such that $L_\varphi$ is
 closed under disjoint $(\log n)^d$-swaps.
 For each fixed $d,n_0\in\NN$ let $n\geq n_0$ and let $w,w'$ be
 strings of length $n$ which witness the violation of the ``closure
 under disjoint $(\log n)^d$-swaps'' property. That is, $w\in L_\varphi$,
 $w'\not\in L_\varphi$, and $w'$ is obtained from $w$ by a disjoint $(\log
 n)^d$-swap operation. Thus, there exist strings $\stringfont{x},\stringfont{u},\stringfont{y},\stringfont{v},\stringfont{z}$ such that
 $w=\stringfont{xuyvz}$ and $w'=\stringfont{xvyuz}$, and for the positions $i,j,i',j'$ of $w$
 just before
 $\stringfont{u},\stringfont{y},\stringfont{v},\stringfont{z}$ the
 following is true for $r:= (\log n)^d$: 
 the neighbourhoods
 $N^{\S_w}_r(i)$, $N^{\S_w}_r(j)$,  $N^{\S_w}_r(i')$,
 $N^{\S_w}_r(j')$ are pairwise disjoint, and 
 $(\NS^{\S_w}_r(i),i)\isom (\NS^{\S_w}_r(i'),i')$ and
 $(\NS^{\S_w}_r(j),j)\isom (\NS^{\S_w}_r(j'),j')$.

\medskip

\noindent
 The overall proof idea is as follows: 
 \begin{enumerate}
  \item
   Choose an appropriate extension
   $\tilde{\sigma}$ of the signature $\sigma_\Sigma$,
  \item 
   modify the formula $\varphi$ into a suitable
   $\FOMODp(\tilde{\sigma}_\ARB)$-formula $\psi(x_1)$ with one free
   variable, and
 \item
   define for each string $w$ and all tuples $\ov p := (i,i',j,j')$ of positions of $w$ which satisfy $0 \leq i < i' < j < j' < \abs{w}$ a 
   $\tilde{\sigma}$-structure $\A:=\A_w^{\ov p}$ with the same universe as $\S_w$, such that
   the positions $a:=i{+}1$ and $a':=i'{+}1$ have disjoint and
   isomorphic $((\log n)^d{-}1)$-neighbourhoods in $\A$ if the
   $(\log n)^d$-neighbourhoods of $i$ and $i'$ and those of $j$ and $j'$ in $\S_{w}$
   are isomorphic and all these $(\log n)^d$-neighbourhoods are pairwise disjoint
 \end{enumerate}
 \noindent
 such that the following is satisfied: 
 \begin{enumerate}[resume]
  \item $\psi(x_1)$ is arb-invariant on $\A$,
  \item $\A\models\psi[a] \; {\iff} \; \S_w\models\varphi$,
  \item $\A\models\psi[a'] \; {\iff} \; \S_{w'}\models\varphi$.
  \end{enumerate}

 \noindent
 Note that $w\in L_\varphi$ and $w'\not\in L_\varphi$ imply that
 $\A\models\psi[a]$ and $\A\not\models\psi[a']$. In combination
 with (3) this shows that the unary query defined by
 $\psi(x_1)$ is \emph{not} weakly Gaifman $((\log n)^d{-}1)$-local on
 the class $\Class$ containing all structures $\A_{w}^{\ov p}$ which we defined above.
 Property (4) is true for each choice of the string $w$ and positions $\ov p$.
 Hence, $\psi(x_{1})$ is arb-invariant on $\Class$.
 Since $d$ can be chosen arbitrarily large, and our choice of the
 formula $\psi(x_1)$ will not depend on $d$, this contradicts the fact that 
 the formula $\psi(x_{1})$ is weakly Gaifman local on $\Class$ with polylogarithmic locality radius
 according to Theorem~\ref{thm:weakGaifmanLocality}.

 \medskip

 \noindent
 The details described in items (1)--(4) are carried out as follows.

 \medskip

 \noindent
 ad (1): \
 Let $\tilde{\sigma}:=\sigma_\Sigma\cup\set{F,X,Y_1,Y_2,Z}$, where $F$
 is a binary relation symbol and $X,Y_1,Y_2,Z$ are unary relation
 symbols. Thus,
 $\tilde{\sigma}=\set{E,F,X,Y_1,Y_2,Z}\cup\setc{P_a}{a\in \Sigma}$.

 \bigskip

 \noindent
 ad (3): \
 Let $\A:=\A_w^{\ov p}$ be the $\tilde{\sigma}$-structure defined as follows (an
 illustration can be found in Figure~\ref{fig:lemma:disjoint-swaps}):
 \begin{itemize}
  \item
   $\A$ has the same universe as $\S_w$, i.e., $A=\set{1,\ldots,|w|}$.
  \item $E^{\A}$ is obtained from $E^{\S_{w}}$ by removing all edges between
    the strings
    $\stringfont{x},\stringfont{u},\stringfont{y},\stringfont{v},\stringfont{z}$, 
            i.e.\\ 
    $ E^{\A}\ :=\ E^{\S_{w}} \setminus \set{ (i,i{+}1),\
      (j,j{+}1),\ (i', i'{+}1),\ (j', j'{+}1)}. $ 
  \item $F^{\A}$ relates the first and last position of $\stringfont{u}$,
    and the first and last position of $\stringfont{v}$, i.e.\\
    $ F^{\A}\ :=\ \set{(i{+}1,j),\ (i'{+}1,j')}.$
  \item $X^{\A}$ marks the last position of $\stringfont{x}$, \ 
    $Y_1^{\A}$ and $Y_2^{\A}$ mark the first and the last position
    of $\stringfont{y}$, \ and \
    $Z^{\A}$ marks the first position of $\stringfont{z}$, i.e.\\ 
    $X^{\A}:=\set{i}$,\ \ $Y_1^{\A}:=\set{j{+}1}$, \ \ 
    $Y_2^{\A}:=\set{i'}$, \ \  $Z^{\A}:=\set{j'{+}1}$. 
  \item for each $a\in\Sigma$, \ $P_a^{\A}$ is identical to
    $P_a^{\S_w}$, i.e.\ \ $P_a^{\A}=P_a^{\S_w}$.
 \end{itemize}
 
 \noindent
 It is easy to see that, if in $\S_w$ the $(\log n)^d$-neighbourhoods of $i$
 and $i'$ are isomorphic and the $(\log n)^d$-neighbourhoods of $j$ and $j'$ are isomorphic and the $(\log n)^d$-neighbourhoods of $i,i',j,j'$ are pairwise disjoint,
 then in $\A$ the $((\log n)^d{-}1)$-neighbourhoods of $a:=i{+}1$ and
 $a':=i'{+}1$ are disjoint and isomorphic.

  \begin{figure}[ht]
    \centering
    \resizebox{250pt}{!}{\includegraphics{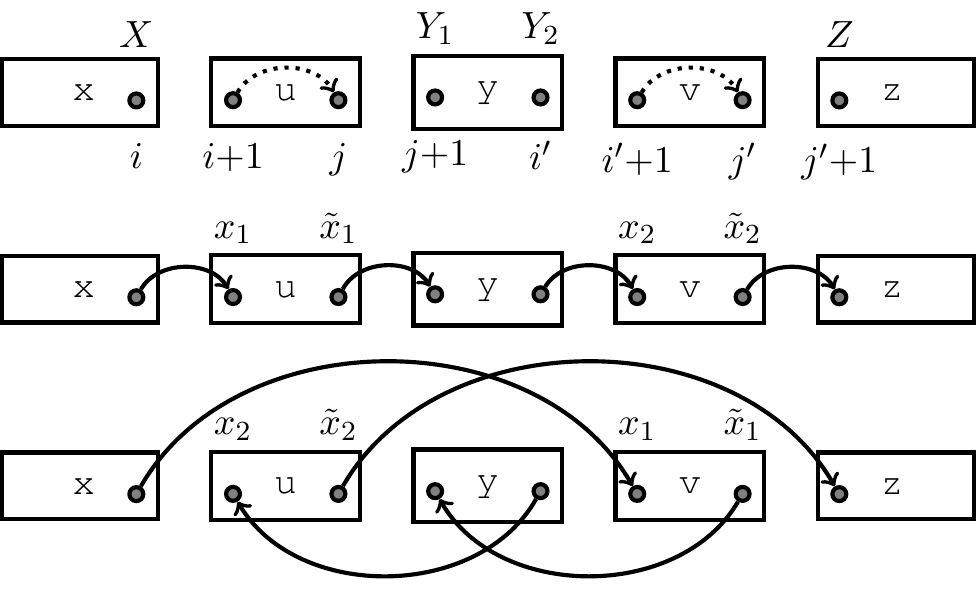}}
    \caption{Structure $\A$ (top) of
      Lemma~\ref{lemma:disjoint-swaps} and the edge relations
      simulated by the formula $\psi(x_1)$ 
      if $x_1$ is assigned the value $a:=i{+}1$ (middle)
      or the value $a':=i'{+}1$ (bottom). $F$-edges in $\A$ are
      depicted by dotted arcs; the $E$-edges in $\A$ are given as successor
      relations within each of the framed boxes. 
    }\label{fig:lemma:disjoint-swaps}
  \end{figure}

 \bigskip

 \noindent
 ad (2) and (4),\,(5),\,(6): \
 We define $\psi(x_1)$ in such a way that, when evaluated in one of the structures $\A:=\A_w^{\ov p}$, it
 does the following:  
  \begin{itemize}
  \item if the variable $x_1$ is assigned the value $a:=i{+}1$, then $\psi(x_1)$ simulates $\varphi$ on $\S_w$ .
  \item if the variable $x_1$ is assigned the value $a':=i'{+}1$, then $\psi(x_1)$ simulates $\varphi$ on $\S_{w'}$ 
  \item if $x_1$ is assigned to a value different from $a,a'$,
     then $\psi(x_1)$ is not satisfied in $\A$.
  \end{itemize}

  \noindent
  The first two items imply items (5) and (6) above.
  Together with the last item and the arb-invariance of $\varphi$ on $\S_{w}$ and $\S_{w'}$, this implies that item (4) is satisfied.\\
 We let \ $\psi(x_1) := \exists \tilde{x}_1 \exists x_2 \exists
 \tilde{x}_2 \; \psi'$, \ where $\psi'$ is
 a conjunction of formulas stating that:
  \begin{itemize}
  \item $F(x_1, \tilde{x}_1)$ holds.
    \ (This ensures that $\psi(x_1)$ is satisfied
    in $\A$ only if $x_1$ is assigned one of the values $a:=i{+}1$
    or $a':=i'{+}1$.)

  \item $x_2 \neq x_1$ and $F(x_2, \tilde{x}_2)$ holds.
  \item
    The formula $\varphi'(x_1, \tilde{x}_1, x_2, \tilde{x}_2)$ is satisfied, where
    $\varphi'$ is obtained from $\varphi$ by
    replacing each atom of the form $E(\mu,\nu)$ with
    the formula $\chi_E(x_1, \tilde{x}_1, x_2, \tilde{x}_2,\mu,\nu)$ defined as follows:\\
    $\chi_E$ is the disjunction of formulas stating that
    \begin{itemize}
    \item $E(\mu,\nu)$ holds
      \
      (i.e., there already exists an $E$-edge from $\mu$ to $\nu$ in $\A$),

    \item $X(\mu)$ holds and $\nu=x_1$
      \
       (i.e., a new $E$-edge from
       the last position of the string $\stringfont{x}$ to the position assigned to
       the variable $x_1$ is introduced)
    \item $\mu =\tilde{x}_1$ and $Y_{1}(\nu)$ holds
      \
      (i.e., a new $E$-edge
      from the position assigned to $\tilde{x}_1$ to the first
      position of the string $\stringfont{y}$
      is introduced)
    \item $Y_{2}(\mu)$ and $\nu=x_2$
     \
      (i.e., a new $E$-edge
      from the last position of the string $\stringfont{y}$ to the
      position assigned to the variable $x_2$ is
      introduced)
    \item $\mu=\tilde{x}_2$ and $Z(\nu)$
     \
      (i.e., a new $E$-edge from the position assigned to the variable $\tilde{x}_2$
      to the first position of the string $\stringfont{z}$ is introduced).
    \end{itemize}

  \noindent It is not difficult to see that $\chi_E$ simulates the edge relation
  of $\S_{w}$ in $\A$ if the variable $x_1$ is assigned the value
  $a:=i{+}1$;
  and it simulates the edge relation of $\S_{w'}$ if the variable $x_1$
  is assigned the value $a':=i'{+}1$ (see
  Figure~\ref{fig:lemma:disjoint-swaps} for an illustration).
  \end{itemize}

  \noindent
  This finishes the construction for items (2) and (4), (5), (6). In summary, the
  proof of Lemma~\ref{lemma:disjoint-swaps} is complete.
\end{proof}
Now that we have proved Theorem~\ref{thm:HanfLocalityOnStrings} for the case $k=0$,
we explain how the general case can be obtained from this.
To this end, we convert a $k$-ary query on strings to a language
over an extended alphabet with the same relevant definability
and locality properties. 
This can be done using a standard technique which encodes variable assignments
in an extended alphabet.

For each alphabet $\Sigma$ and each $k\geq 1$, we let $\Sigma_{\textup{var}(k)} := \Sigma \times 2^{[k]}$.
The subsets of $[k]$ are used to encode an assignment of $k$ variables
to the positions of a string over the alphabet $\Sigma$.
We let $L_{\textup{assign}(k)}$ denote the language of strings $w\in\Sigma_{\textup{var}(k)}^{+}$ 
where for each number $i\in [k]$ there is a unique position $\operatorname{occ}_{i}(w) \in \set{1,\ldots, \abs{w}}$
such that the label of $\occ_{i}(w)$ is $(a,X) \in \Sigma_{\textup{var}(k)}$ with $i\in X$.
Note that this language is clearly $\FO[\sigma_{\Sigma_{\textup{var}(k)}}]$-definable.
For each string $w\in L_{\textup{assign}(k)}$, we let $\tilde w\in \Sigma^{+}$
denote the string where each symbol $(a,X)$ is replaced by $a$.
Furthermore, we let $\ov{\occ_{k}}(w) := (\occ_{0}(w), \ldots, \occ_{k-1}(w))$.

\begin{lemma}\label{lem:string-query-arity-reduction}
  Let $\Sigma$ be a finite alphabet and let $k\in \NNpos$. With each $k$-ary query $q$ on
  $\ClassStrings$ we associate the language
  \[ A_{q} \ := \  \setc{w\in L_{\textup{assign}(k)}}{\ov{\occ}_{k}(w) \in q(S_{\tilde w})}.\] 
  Then 
  \begin{enumerate}
  \item\label{item:definability-q-to-Aq} For each $p\in \NNpos$, if $q$ is definable in $\ordinvFOMODp(\sigma_\Sigma)$
    on $\ClassStrings$, then $A_q$ is definable in $\ordinvFOMODp(\sigma_{\Sigma_{\textup{var}(k)}})$.
  \item\label{item:locality-q-to-Aq} The query $q$ is Hanf $f(n)$-local on $\ClassStrings$ iff the language $A_{q}$ is Hanf $f(n)$-local.
  \end{enumerate}
\end{lemma}
\begin{proof}\mbox{}
  \begin{enumerate}[label=Ad (\arabic*):]
  \item[Ad~\eqref{item:definability-q-to-Aq}:]

    Let $\phi(\ov x)$ be a formula of $\ordinvFOMODp(\sigma_\Sigma)$ defining $q$, where $\ov x := (x_{0}, \ldots, x_{k-1})$.
    For each $i\in [k]$, let $\psi_{\occ,i}(x) := \bigvee_{a\in \Sigma, X \in 2^{[k]}, i\in X} P_{(a,X)}(x)$.
  That is, the formula states that $i$ occurs in the set in the second component of the label at position $x$.
    Let
    \[ \hat \phi := \exists x_{0} \dotsb \exists x_{k-1} \bigwedge_{i\in [k]} \psi_{\occ,i}(x_{i}) \land \phi', \]
    where $\phi'$ is obtained from $\phi$ by replacing each occurrence of a relation symbol $P_{a}(x)$, for $a\in \Sigma$, by the formula $\bigvee_{X \in 2^{[k]}} P_{(a,X)}(x)$
    which states for a position $x$ of $\S_{w}$, for each string $w\in \Sigma_{\textup{var}(k)}^{+}$, that the label at position $x$ in $\tilde w$ is $a$.
    Let $\psi$ be an $\FO$-sentence which defines $L_{\textup{assign}(k)}$ on $\ClassStrings[\Sigma_{\textup{var}(k)}]$.
    It is straightforward to verify that $\hat \phi \land \psi$ is a formula of $\ordinvFOMODp(\sigma_{\Sigma_{\textup{var}(k)}})$ which defines the language $A_{q}$.

  \item[Ad~\eqref{item:locality-q-to-Aq}:]
    Note that
    $\S_{u} \Hanfequiv_{f(n)} \S_{v}$ iff
    $(\S_{\tilde u}, \ov{\occ}_{k}(u)) \Hanfequiv_{f(n)} (\S_{\tilde v}, \ov{\occ}_{k}(v))$,
    for all $u,v \in L_{\textup{assign}(k)}$ of length $n$.
    More concretely, a bijection $\beta$
    witnesses the first statement iff it witnesses the second statement.
    This follows from the fact that  $\beta$ preserves
    the second component of the labelling of
    $\S_{u}$ (i.e. the label of $\beta(x)$ is the same as the label of $x$, for all elements
    $x\in\set{1, \ldots,n}$) iff $\beta(\occ_{i}(u)) = \occ_{i}(v)$, for all $i\in [k]$,
    and that the label of each position of $\tilde u$ and $\tilde v$
    is the first component of the labelling of that position in $u$ and $v$.

    Note also that $L_{\textup{assign}(k)}$ is $0$-local since its definition depends only on the labelling.

    First, we show that locality of $q$ implies locality of $A_{q}$.
    Suppose that $\S_{u} \Hanfequiv_{f(n)} \S_{v}$, for 
    $u,v \in \Sigma_{\textup{var}(k)}^{n}$.
    We have $u\in L_{\textup{assign}(k)}$ iff $v\in L_{\textup{assign}(k)}$.
    If $u,v \notin L_{\textup{assign}(k)}$, then $u,v \notin A_{q}$.
    Suppose that $u,v \in L_{\textup{assign}(k)}$.
    Then $(\S_{\tilde u}, \ov{\occ}_k(u)) \Hanfequiv_{f(n)} (\S_{\tilde v}, \ov{\occ}_k(v))$
    and we obtain that $\ov{\occ}_k(u) \in q(\S_{\tilde u})$ iff $\ov{\occ}_k(v) \in q(\S_{\tilde v})$, by $f(n)$-locality of $q$, and so $u \in A_{q}$ iff $v\in A_{q}$.

    Now we show that locality of $A_{q}$ implies locality of $q$.
    Let $u,v \in \Sigma^{n}$ and $\ov a = (a_{0}, \ldots, a_{k-1}), \ov b := (b_{0}, \ldots, b_{k-1})\in \set{1,\ldots,n}^{k}$
    be such that $(\S_{u},\ov a) \Hanfequiv_{f(n)} (\S_{v}, \ov b)$.
    There exist $u',v' \in L_{\textup{assign}(k)}$ such that $u = \tilde u'$ and $v = \tilde v'$ and such that $a_{i} = \occ_{i}(u')$ and $b_{i} = \occ_{i}(v')$, for each $i\in [k]$.  Then $(\S_{u},\ov a) = (\S_{\tilde u'},\ov{\occ}_k(u'))$ and $(\S_{v},\ov b) = (\S_{\tilde v'},\ov{\occ}_k(v'))$. Since $(\S_{u},\ov a) \Hanfequiv_{f(n)} (\S_{v}, \ov b)$, we also have
    $\S_{u'} \Hanfequiv_{f(n)} S_{v'}$. By Hanf $f(n)$-locality of $A_{q}$, this means that
    $u' \in A_{q}$ iff $v' \in A_{q}$. Hence, $\ov{\occ}_k(u') \in q(\S_{\tilde u'})$ iff $\ov{\occ}_k(v') \in q(\S_{\tilde v'})$. That is, $\ov a \in q(\S_{u})$ iff $\ov b \in q(\S_{v})$
    and hence $q$ is Hanf $f(n)$-local on $\ClassStrings$.\qedhere
  \end{enumerate}
\end{proof}

\noindent The proof of Theorem~\ref{thm:HanfLocalityOnStrings} can now be finished as follows.
Suppose that 
there is a $k$-ary query $q$, for some $k>0$, which is definable in
$\arbinvFOMODp(\sigma_\Sigma)$ on $\ClassStrings$,
for some odd prime power $p$ and some finite alphabet $\Sigma$,
such that 
there is no $c\in\NN$ for which $q$ is Hanf $(\log n)^c$-local on $\ClassStrings$.
By the previous lemma, we obtain a language $A_{q}$ which is definable in 
$\arbinvFOMODp(\sigma_{\Sigma_{\textup{var}(k)}})$ and which is also not Hanf $(\log n)^{c}$-local. But we have already proved that this is impossible.

\section{Conclusion}\label{section:conclusion}
We have introduced a new notion of locality called \emph{shift locality} which generalises
Niemistö's notion of alternating locality with constant locality radius \cite{Niemistoe-PhD} and Libkin's notion of weak Gaifman locality \cite{Libkin2004}.
We have presented a comprehensive picture of the locality of $\arbinvFOMODp$ for prime powers $p$.

We have also investigated notions of locality on string structures.
Here some natural questions remain open.
We have shown that there is an $\ordinvFOMODnum{2}$-definable language which is not Hanf local and hence not $\FOMOD{q}$-definable for any modulus $q$. It would be interesting to understand the expressive power of $\ordinvFOMODnum{2}$ on strings in more detail. Is there a decidable algebraic characterisation of the $\ordinvFOMODnum{2}$-definable languages? Is there a logic with the same expressive power as $\ordinvFOMODnum{2}$ but with an effective syntax?

We have derived Hanf locality on string structures from weak Gaifman locality. To accomplish this, we have used the characterisation of Hanf locality in terms of disjoint swap operations.
The origin of these swap operations goes back to \cite{BenediktSegoufin2009b} where similar operations are, among other things, used to obtain a new proof of an algebraic characterisation of first-order logic on string structures
from \cite{BP89}. As a first step towards proving an algebraic characterisation of $\ordinvFOMODnum{2}$, we believe that one could define a variant of Hanf locality whose relation to shift locality is similar to the relation of Hanf locality and weak Gaifman locality that we established.
Then one could try to characterise this notion of locality by a variant of the disjoint swap operations.

\bibliographystyle{abbrv} 
 
\end{document}